\documentclass[letterpaper, 10 pt, conference]{ieeeconf}  % Comment this line out if you need a4paper

\IEEEoverridecommandlockouts                              % This command is only needed if 
\usepackage{amsfonts}
\usepackage{amstext}
\usepackage{mathtools}
\usepackage{cases}
\usepackage{enumerate}
\usepackage{subfigure}
\usepackage{amssymb}
\usepackage{float}
\usepackage{comment}
\usepackage[dvipsnames]{xcolor}
\usepackage{hyperref}
\usepackage{matlab-prettifier}
\usepackage{multirow}

\usepackage{algorithm}
\usepackage{algpseudocode}
\usepackage{etoolbox}
\usepackage{booktabs}

\makeatletter
\newenvironment{breakablealgorithm}
  {%
   \begin{center}
     \refstepcounter{algorithm}%
     \hrule height.8pt depth0pt \kern2pt
     \renewcommand{\caption}[2][]{%
       {\raggedright\textbf{\ALG@name~\thealgorithm} ##2\par}%
     }
  }
  {%
     \kern2pt\hrule\relax
   \end{center}
  }
\makeatother

\makeatletter
\newcommand{\StateWrap}[1]{%
    \State \parbox[t]{\dimexpr\linewidth-\ALG@thistlm\relax}{%
        \strut#1\strut
    }%
}
\makeatother

\newtheorem{thm}{\textbf{\text{Theorem}}}
\newtheorem{lem}{\textbf{\text{Lemma}}}

\newtheorem{Definition}{\textbf{\text{Definition}}}

\newtheorem{rmk}{\textbf{\text{Remark}}}

\newcommand{\ie}{\textit{i.e.}}

\title{\LARGE \bf Nonlinear Observer Design for Visual-Inertial Odometry}
\author{Mouaad Boughellaba, Abdelhamid Tayebi, James R. Forbes, and Soulaimane Berkane  % <-this % stops a space
\thanks{This work was supported by the National Sciences and Engineering Research Council of Canada (NSERC), under the grants NSERC-DG RGPIN 2020-06270 and NSERC-DG RGPIN-2020-04759, and by Fonds de recherche du Qu\'ebec (FRQ).} 
\thanks{M. Boughellaba and A. Tayebi are with the Department of Electrical Engineering, Lakehead University, Thunder Bay, ON P7B 5E1, Canada \tt\small \{mboughel,atayebi\}@lakeheadu.ca.} 
\thanks{S. Berkane is with the Department of Computer Science and Engineering, University of Quebec in Outaouais, Gatineau, QC, Canada. {\tt\small Soulaimane.Berkane@uqo.ca}}
\thanks{James R. Forbes is with the Department of Mechanical Engineering, McGill University, Montreal, WC, Canada. \texttt{james.richard.forbes@mcgill.ca}}
}

\begin{document}

\maketitle
\thispagestyle{empty}
\pagestyle{empty}

%%%%%%%%%%%%%%%%%%%%%%%%%%%%%%%%%%%%%%%%%%%%%%%%%%%%%%%%%%%%%%%%%%%%%%%%%%%%%%%%
\begin{abstract}
This paper addresses the problem of Visual-Inertial Odometry (VIO) for rigid body systems evolving in three-dimensional space under a bias-free inertial measurement unit (IMU) model. We formulate the VIO problem on the matrix Lie group \(\mathbb{SE}_{3+n}(3)\), which jointly represents the pose, velocity, gravity, and landmark positions. Building upon this formulation, we design an almost globally asymptotically stable nonlinear geometric observer that tightly integrates data from an Inertial Measurement Unit (IMU) and visual sensors. Unlike conventional Extended Kalman Filter (EKF)-based estimators that rely on local linearization and thus ensure only local convergence, the proposed observer achieves almost global stability through the decoupling of the rotational and translational dynamics. A globally exponentially stable Riccati-based translational observer and an almost globally input-to-state stable attitude observer are designed such that the overall cascaded observer enjoys almost global asymptotic stability.
Under the stated observability conditions, the reduced-attitude and translational estimation errors converge to zero from almost all initial conditions. Global translation and rotation about gravity remain unobservable. The proposed observer is evaluated through numerical simulations and experiments on all eleven sequences of the \textit{EuRoC MAV} dataset using stereo and monocular configurations.
\end{abstract}

\section{Introduction}
Accurate estimation of vehicles' motion in GPS-denied environments is a fundamental problem in robotics. Among the available approaches, Visual--Inertial Odometry (VIO) has emerged as a promising solution by tightly fusing information from a camera and an inertial measurement unit (IMU). While VIO shares similarities with Simultaneous Localization and Mapping (SLAM), its primary objective is different: VIO focuses on estimating the ego-motion of the sensing platform without explicitly building or maintaining a persistent map of the surrounding environment. VIO also extends Visual Odometry (VO) by incorporating IMU measurements to complement visual information and improve motion estimation. Cameras provide rich geometric information about the environment but suffer from scale ambiguity and sensitivity to illumination changes. Pure IMU-based odometry (relying on model integration) provides state estimates that drift over time due to sensor noise and biases. By combining these complementary sensing modalities, VIO achieves robust, accurate, and drift-resilient state estimation \cite{Lim2020ARO}.

State-of-the-art solutions for VIO are predominantly based on two paradigms: filter-based approaches, typically using the extended Kalman filter (EKF), and optimization-based approaches, often realized through sliding-window optimization \cite{Gui_2015}. In EKF-based VIO, the system state is recursively propagated using IMU (gyro and accelerometer) measurements and corrected using visual observations.
However, filter-based methods rely on local linearizations, which can lead to inconsistencies due to the accumulation of linearization errors. In contrast, optimization-based VIO formulates the motion estimation as a nonlinear least-squares problem, where recent poses and landmark positions are jointly optimized within a sliding time window. This category achieves higher accuracy, but it is computationally more demanding. As a result, the choice between EKF-based and optimization-based frameworks often reflects a trade-off between computational efficiency and estimation accuracy.

In this work, we propose a nonlinear geometric observer for Visual--Inertial Odometry whose main contribution is an almost global asymptotic stability guarantee under a bias-free
IMU model.
Specifically, we extend the matrix Lie group representation on \(\mathbb{SE}_{2+n}(3)\) by incorporating an auxiliary gravity state, yielding a formulation on \(\mathbb{SE}_{3+n}(3)\) that jointly represents orientation, position, velocity, gravity, and landmark positions. This formulation enables the fusion of IMU and visual measurements without relying on local linearization of the observer error dynamics.
The key idea is to exploit the natural geometry of the underlying state space to decouple the rotational and translational dynamics by introducing an auxiliary system related to the gravity vector. This leads to a cascaded observer structure. The translational subsystem is modeled as a linear time-varying system, for which a Riccati-based design ensures global exponential stability, while the rotational dynamics are governed by an almost globally input-to-state stable attitude observer. This design overcomes the limitations of conventional EKF- and optimization-based methods, which typically provide only local convergence guarantees or lack provable stability results. Through this development, we provide a rigorous geometric framework for VIO estimation that combines theoretical soundness with practical robustness, and whose effectiveness is demonstrated through both numerical simulations and experimental validation on the \textit{EuRoC MAV} dataset. The extended state is estimated up to the unobservable directions of VIO.

\section{RELATED WORK}

\subsection{EKF- and optimization-based VIO}

A landmark contribution in filter-based VIO was the Multi-State Constraint Kalman Filter (MSCKF) introduced by Mourikis and Roumeliotis \cite{Mourikis_ICRA2007}. Unlike conventional EKF formulations that explicitly included 3D landmark positions in the state vector, the MSCKF introduced the idea of enforcing multi-view feature constraints without directly estimating landmark locations. This innovation drastically reduced the computational complexity from quadratic to linear in the number of visual features while maintaining estimator consistency, thereby enabling real-time operation on resource-constrained platforms. Following its introduction, several extensions of the MSCKF were proposed to enhance robustness and accuracy. For example, Li and Mourikis \cite{Mingyang_2013} demonstrated that the MSCKF algorithm achieves better accuracy and consistency than EKF-based VIO algorithms, due to its less restrictive probabilistic assumptions and delayed linearization, and further introduced the MSCKF 2.0 algorithm, which incorporates a closed-form expression for the IMU error-state transition matrix and fixed linearization states to ensure proper observability, while simultaneously performing online camera-to-IMU calibration. Beyond the classical MSCKF framework, further innovations focused on addressing the limitations of global reference frames and improving robustness in practical scenarios. Robocentric formulations were introduced to represent the state with respect to a moving local frame rather than a fixed global frame, improving numerical stability and estimator consistency during aggressive maneuvers \cite{Huai_IJRR2022,Huai_IROS2018}. Around the same time, Bloesch et al. developed ROVIO \cite{Bloesch_IJRR2017,Bloesch_IROS2015}, an iterated EKF-based algorithm that departed from conventional feature-based approaches by employing a direct photometric error formulation, where image pixel intensities were integrated directly into the filter update. This approach reduced computational cost while maintaining robustness under challenging visual conditions. More recently, Geneva \textit{et al.} introduced OpenVINS \cite{Geneva_ICRA2020}, an EKF-based framework that combines the MSCKF \cite{Mourikis_ICRA2007} with First-Estimates Jacobian (FEJ) treatments \cite{Huang_IJRR2010,Huang_ISER_2009} to achieve state-of-the-art performance. Its main contribution lies in offering a modular, open-source, and well-documented platform for EKF-based VIO research, providing a valuable baseline for comparing against optimization-based methods.

In addition to filtering-based approaches, optimization-based VIO methods have also emerged in the literature. A seminal contribution in this direction is OKVIS \cite{Leutenegger_OKVIS}, which introduced a tightly coupled, keyframe-based estimator capable of supporting both monocular and stereo configurations. Building upon advances in visual SLAM, Qin \textit{et al.} proposed VINS-Mono \cite{Qin_TOR2018}, which extended the sliding-window formulation with efficient marginalization, loop closure, and global pose graph optimization, enabling drift reduction and high accuracy in large-scale environments. Complementing these developments, Delmerico and Scaramuzza \cite{Scaramuzza_ICRA2018} benchmarked a wide range of state-of-the-art VIO systems and showed that algorithms optimized primarily for speed and CPU efficiency often suffered from reduced accuracy and struggled on challenging datasets or resource-limited hardware.

\subsection{Invariant and Equivariant Observers for VIO}

Over the past decade, a growing body of work has investigated invariant and equivariant observers as alternatives to classical EKF- and optimization-based VIO, motivated by their ability to exploit Lie group symmetries to improve consistency and robustness. The two principal families are the Invariant EKF (IEKF) \cite{barrau2018invariant} and the more recent Equivariant Filter (EqF) \cite{Pieter_TAC2023}. Barrau and Bonnabel introduced the extended special Euclidean group $\mathbb{SE}_2(3)$, which enables an exact linearization of IMU error dynamics when biases are known, thereby overcoming limitations of the standard on-manifold EKF and the multiplicative EKF widely used in SLAM and VIO. Their general invariant-filtering analysis \cite{barrau2016invariant} establishes exact log-linear error propagation for group-affine systems and local asymptotic stability of the IEKF under the stated conditions. Building on this framework, their IEKF-SLAM formulation \cite{barrau_arxiv2016} uses extended special Euclidean groups to preserve the unobservable global reference-frame directions and establish associated consistency properties. These consistency results do not establish almost global convergence; the convergence guarantee provided by the general IEKF analysis is local.

Subsequent studies analyzed observability properties of IEKF-based formulations \cite{Zhang_RAL2017} and compared their performance to conventional EKF-based approaches. Building on these foundations, Wu \textit{et al.} \cite{Wu_IROS2017} combined the MSCKF framework with invariant filtering principles to design a consistent Invariant MSCKF, addressing the growing inconsistency observed in the original MSCKF under Monte Carlo simulations. Extensions have also been made using unscented filters on Lie groups \cite{Brossard_2018}, as well as coupling IEKF with FEJ techniques to further enhance accuracy and consistency on real datasets such as TUM-VI \cite{Yang_RAL2022}.

The EqF framework extends these ideas to systems whose state spaces admit transitive Lie group actions, including cases where the symmetry group has a higher dimension than the state space and the IEKF cannot be directly applied \cite{Pieter_TAC2023}. It also provides a systematic way to exploit measurement equivariance to reduce output approximation errors. Building on this foundation, the authors in \cite{Pieter_TR2023} proposed EqVIO, a consistent equivariant filter for tightly coupled VIO. Its symmetry combines $\mathbb{SE}_2(3)$ for the navigation states with scaled-rotation groups $\mathbb{SOT}(3)$ for the landmarks, enabling exact bias-free navigation-error propagation and a higher-order equivariant approximation of the visual output. The formulation also includes online estimation of IMU biases and camera extrinsics. Their experimental evaluation on the EuRoC and UZH FPV datasets shows that the proposed method outperforms state-of-the-art VIO algorithms in both speed and accuracy. However, that work does not establish an almost global asymptotic stability result for the complete observer.

A common rigid transformation of the robot pose and landmarks leaves relative landmark measurements unchanged, which motivates the quotient-space formulation of SLAM in~\cite{Mahony_SLAM_CDC2017}. On this space, the observer in~\cite{Mahony_AJC2021} achieves global asymptotic and local exponential convergence. In the VIO setting considered here, the known gravity direction limits the ambiguity to global translation and rotation about gravity. The distinction of the present work lies in the observer design and its stability analysis, rather than in the general extended-group construction. In particular, the auxiliary gravity state and innovation design make the translational error dynamics exactly linear and independent of the attitude error, including the continuous visual correction. This property enables the cascaded almost global asymptotic stability analysis under the stated assumptions, including a bias-free IMU model.

The remainder of this paper is structured as follows: Section~\ref{s2} provides the preliminaries and notation and introduces the extended matrix Lie group $\mathbb{SE}_{3+n}(3)$ used in our formulation. Section~\ref{s3} defines the problem of visual--inertial odometry. Section~\ref{s4} presents the proposed VIO observer and its observability and stability analyses. Section~\ref{s6} reports simulation results, while Section~\ref{s7} presents implementation details and experimental results on the \textit{EuRoC MAV} dataset.
Finally, Section~\ref{s8} concludes the paper.

\section{Preliminaries}\label{s2}
\subsection{Notations}
The set of real numbers and the $n$-dimensional Euclidean space are denoted by $\mathbb{R}$ and $\mathbb{R}^n$, respectively. The set of unit vectors in $\mathbb{R}^n$ is defined as $\mathbb{S}^{n-1}:=\{x\in \mathbb{R}^n~|~x^\top x =1\}$. Given two matrices $A, B$ $\in \mathbb{R}^{m\times n}$, their Euclidean inner product is defined as $\langle \langle A,B \rangle \rangle=\text{tr}(A^\top B)$. The Euclidean norm of a vector $x \in \mathbb{R}^n$ is defined as $||x||=\sqrt{x^\top x}$, and the Frobenius norm of a matrix $A \in \mathbb{R}^{n\times n}$ is given by $||A||_F=\sqrt{\langle \langle A,A \rangle \rangle}$. The identity matrix is denoted by $I_n \in \mathbb{R}^{n \times n}$. The attitude of a rigid body is represented by a rotation matrix $R$ which belongs to the special orthogonal group $\mathbb{SO}(3):= \{ R\in \mathbb{R}^{3\times 3} | \hspace{0.1cm}\text{det}(R)=1, R^\top R=I_3\}$. The tangent space of the compact manifold $\mathbb{SO}(3)$ is given by $T_R\mathbb{SO}(3):=\{R \hspace{0.1cm}\Omega \hspace{0.2cm} | \hspace{0.2cm} \Omega \in \mathfrak{so}(3)\}$, where $\mathfrak{so}(3):=\{ \Omega \in \mathbb{R}^{3\times 3} | \Omega^\top=-\Omega\}$ is the Lie algebra of the matrix Lie group $\mathbb{SO}(3)$. The map $[.]_{\times}: \mathbb{R}^3 \rightarrow \mathfrak{so}(3)$ is defined such that $[x]_\times y=x \times y$, for any $x,y \in \mathbb{R}^3$, where $\times$ denotes the vector cross product on $\mathbb{R}^3$. The angle-axis parameterization of $\mathbb{SO}(3)$ is given by $\mathcal{R}(\theta, v):=I_3+\sin\hspace{0.05cm}\theta \hspace{0.2cm}[v]_{\times} + (1-\cos\hspace{0.05cm}\theta)([v]_{\times})^2$, where $v\in \mathbb{S}^2$ and  $\theta \in \mathbb{R}$ are the rotation axis and angle, respectively. The Kronecker product of two matrices $A$ and $B$ is denoted by $A \otimes B$. 
The \textit{vectorization map}
$\operatorname{vec}:
\mathbb{R}^{m\times n}\rightarrow\mathbb{R}^{mn}$
is defined by
\begin{equation}
    \operatorname{vec}
    \left([x_1~x_2~\hdots~x_n]\right)
    =
    [x_1^\top,x_2^\top,\hdots,x_n^\top]^\top,
\end{equation}
where each $x_i\in\mathbb{R}^m$ for
$i\in\{1,2,\hdots,n\}$. Notice that the map \( \operatorname{vec}() \) stacks the columns of a matrix vertically into a single vector. The \textit{inverse map}, denoted \( ()^\wedge : \mathbb{R}^{mn} \rightarrow \mathbb{R}^{m \times n} \), is defined by
\begin{equation}
    \left([x_1^\top, x_2^\top, \hdots, x_n^\top]^\top\right)^\wedge = [x_1~x_2~\hdots~x_n],
\end{equation}
which reconstructs the original matrix by reshaping the stacked vector back into its column-wise form. Let $\pi : \mathbb{R}^{n} \setminus\{0\} \to \mathbb{R}^{n\times n}$ denote the orthogonal projection operator, which will be used throughout this paper, defined by  
\begin{equation}
    \pi(x) = I_n - \frac{x\,x^\top}{||x||^2}, 
    \label{eq:projection}
\end{equation}
where $x \in \mathbb{R}^n\setminus\{0\}$. The matrix $\pi(x)$ is an orthogonal projection that maps any nonzero vector in $\mathbb{R}^n$ onto the subspace orthogonal to $x$. 
It is bounded and positive semidefinite and satisfies $\pi(x)y = 0_{n\times 1}$ when $x$ and $y$ are collinear. The column and row vectors of ones are denoted by $\mathbf 1_{n\times1}$ and $\mathbf 1_{1\times n}$, respectively.

\subsection{Matrix Lie Group for VIO}

Extended matrix Lie groups combining orientation with multiple vector components have been used in inertial navigation and SLAM \cite{barrau_arxiv2016,barrau2016invariant,Mahony_SLAM_CDC2017}.
Following these constructions, we introduce $\mathbb{SE}_{3+n}(3)$ to represent the VIO state augmented with an auxiliary gravity vector.
This group belongs to the existing family of extended special Euclidean groups and supports the observer design and error decomposition developed below.
In particular, we define $\mathbb{SE}_{3+n}(3)\subset \mathbb{R}^{(6+n)\times(6+n)}$
as follows:
\begin{align}
    \mathbb{SE}_{3+n}(3):=\{
    &X=\mathcal{M}(R,x_1,x_2,x_3,\bold{x_L}):
    \nonumber\\
    &R\in \mathbb{SO}(3),\;
    x_1,x_2,x_3\in\mathbb{R}^3,\;
    \bold{x_L}\in\mathbb{R}^{3\times n}
    \},
\end{align}
where the map
$\mathcal{M}:\mathbb{SO}(3)\times\mathbb{R}^3
\times\mathbb{R}^3\times\mathbb{R}^3
\times\mathbb{R}^{3\times n}
\rightarrow\mathbb{R}^{(6+n)\times(6+n)}$
is defined as follows:
\begin{equation*}
    \mathcal{M}(R,x_1,x_2,x_3,\bold{x_L}):=\left[
    \begin{array}{c|cccc}
        R & x_1 & x_2 & x_3 & \bold{x_L}\\
        \hline
        0_{(3+n)\times3}
        & \multicolumn{4}{c}{I_{3+n}}
    \end{array}\right].
\end{equation*}
This representation extends the classical special Euclidean group $\mathbb{SE}(3)$ with additional vector components. In the VIO formulation considered below, $x_1$, $x_2$, and $x_3$ represent position, velocity, and gravity, respectively, while $\bold{x_L}$ contains the landmark positions.
It is straightforward to verify that $\mathbb{SE}_{3+n}(3)$ indeed forms a matrix Lie group. The group inverse is obtained as
\[
    X^{-1}=\mathcal{M}\bigl(R^\top,-R^\top x_1,-R^\top x_2,-R^\top x_3,-R^\top\bold{x_L}\bigr).
\]
Furthermore, the group multiplication, defined as standard matrix multiplication, is closed, \ie, for any $X_1,X_2\in \mathbb{SE}_{3+n}(3)$, the product $X_1X_2$ is also an element of $\mathbb{SE}_{3+n}(3)$.
The identity element is simply the identity matrix $I_{6+n}$, \ie, $XX^{-1}=X^{-1}X=I_{6+n}$, for every $X\in \mathbb{SE}_{3+n}(3)$.
The Lie algebra of the matrix Lie group $\mathbb{SE}_{3+n}(3)$, denoted by $\mathfrak{se}_{3+n}(3) \subset \mathbb{R}^{(6+n)\times(6+n)}$, is given as follows:
\begin{align}
    \mathfrak{se}_{3+n}(3):=\{
    &V=\mathcal{V}(\Omega,x_1,x_2,x_3,\bold{x_L}):
    \nonumber\\
    &\Omega\in\mathfrak{so}(3),\;
    x_1,x_2,x_3\in\mathbb{R}^3,\;
    \bold{x_L}\in\mathbb{R}^{3\times n}
    \},
\end{align}
where the map
$\mathcal{V}:\mathfrak{so}(3)\times\mathbb{R}^3
\times\mathbb{R}^3\times\mathbb{R}^3
\times\mathbb{R}^{3\times n}
\rightarrow\mathbb{R}^{(6+n)\times(6+n)}$
is defined as follows:
\begin{equation*}
    \mathcal{V}(\Omega,x_1,x_2,x_3,\bold{x_L})
    :=
    \left[
    \begin{array}{c|cccc}
        \Omega & x_1 & x_2 & x_3 & \bold{x_L}\\
        \hline
        0_{(3+n)\times3}
        & \multicolumn{4}{c}{0_{3+n}}
    \end{array}
    \right].
\end{equation*}
Here, $\Omega \in \mathfrak{so}(3)$ denotes a skew-symmetric matrix encoding the infinitesimal rotation, while $x_1, x_2, x_3,$ and $\bold{x_L}$ represent infinitesimal translations associated with the inertial and landmark states.
\begin{rmk}
  The matrix Lie group $\mathbb{SE}_{3+n}(3)$ offers a compact and elegant representation that unifies rigid-body motions with additional translational states within a single algebraic framework. This formulation provides the foundation for developing our geometric observer for VIO in the subsequent sections.  
\end{rmk}

\subsection{Uniform Observability and Continuous Riccati Equation}
Consider the following linear time-varying (LTV) system:
\begin{align}
    \dot{\mathbf{x}}&=A(t)\,\mathbf{x}+B(t)\, \mathbf{u},\label{ltv_1}\\
    \mathbf{y}&= C(t)\, \mathbf{x},\label{ltv_2}
\end{align}
where \(\mathbf{x}(t)\in \mathbb{R}^{n}\), \(\mathbf{u}(t) \in \mathbb{R}^\ell\), and \(\mathbf{y}(t) \in \mathbb{R}^m\) denote the state, input, and output vectors of the system, respectively. The matrices $A(t)$, $B(t)$, and $C(t)$, with appropriate dimensions, are assumed to be continuous and bounded for all \( t \geq 0 \). 
%The following definition presents a key concept associated with uniform observability.

It is well known that the ability to reconstruct the system state from input–output data is closely related to the observability properties of the system. For time-invariant systems, the classical notion of observability, established via the rank condition of the observability matrix, provides a convenient test. In contrast, for time-varying systems, observability must be formulated in a uniform sense, requiring that sufficient excitation be maintained over finite time intervals. This requirement is formalized by the following definition \cite{Chen1999}.  
\begin{Definition}
    The pair \( (A(t), C(t)) \) is said to be uniformly observable if there exist constants \( \delta > 0 \) and \( \mu > 0 \) such that $\forall t \geq 0$:
    \begin{equation}
    W(t, t + \delta) = \frac{1}{\delta} \int_t^{t + \delta} \Phi(\tau, t)^\top C(\tau)^\top C(\tau) \Phi(\tau, t) \, d\tau \geq \mu I_n,
    \end{equation}
    where \( \Phi(t, \tau) \) is the state transition matrix corresponding to \( A(t) \), defined as the solution to \(\frac{d}{dt}{\Phi}(t, \tau)=A(t)\, \Phi(t, \tau)\) with \(\Phi(t, t)=I_n\). The matrix $ W(t, t + \delta)$ is referred to as the \textit{Observability Gramian (OG)} of the system \eqref{ltv_1}–\eqref{ltv_2}. Note that the matrix \( W(t, t + \delta) \) is naturally upper bounded by some constant since the matrices \( A(t) \) and \( C(t) \) are assumed to be bounded for all $t\geq0$.
\end{Definition}
The observability property plays a fundamental role in the behavior of filtering and state estimation algorithms. In particular, it is directly connected to the solvability of the Continuous Riccati Equation (CRE), which arises as a key component in optimal filtering problems. Consider the CRE
\begin{align}
    \dot{P}(t) &= A(t)P(t) + P(t)A(t)^\top \nonumber\\
    &\quad - P(t)C(t)^\top Q^{-1}(t)C(t)P(t)+ V(t),\label{CRE}
\end{align}
where \( P(0) \in \mathbb{R}^{n \times n} \) is a symmetric positive-definite matrix, and \( V(t) \in \mathbb{R}^{n \times n} \), \( Q(t) \in \mathbb{R}^{m \times m} \) are uniformly positive-definite and bounded matrices.

Suppose that $A(t)$ and $C(t)$ are continuous and uniformly bounded, and that the symmetric matrices $V(t)$ and $Q(t)$ satisfy $v_m I_n\leq V(t)\leq v_M I_n$ and $q_m I_m\leq Q(t)\leq q_M I_m$, for some positive constants $v_m,v_M,q_m,q_M$. If $(A(t),C(t))$ is uniformly observable, then, for any symmetric positive-definite initial condition $P(0)$, the solution of \eqref{CRE} exists for all $t\geq0$ and satisfies
\begin{equation}\label{eq:riccati_bounds}
    p_m I_n\leq P(t)\leq p_M I_n,\qquad t\geq0,
\end{equation}
for some constants $0<p_m\leq p_M<\infty$ \cite{Bucy1972}.

\section{Problem Statement}\label{s3}
Let $\{\mathcal{I}\}$ and $\{\mathcal{B}\}$ be the inertial frame and the body-fixed frame attached to the center of mass of a rigid body, respectively. Consider the following dynamics of a rigid body and a set of $n$ static landmarks:
\begin{align}
    \dot{R}&=R[\omega^\mathcal{B}]_\times \label{equ:dynamics1}\\
    \dot{p}&=v\label{equ:dynamics12}\\
    \dot{v}&=g+
    Ra^\mathcal{B}\label{equ:dynamics11}\\
    %\dot{g}&=0\\
    \dot{p}_i&=0,\label{equ:dynamics4}
\end{align}
where $R \in \mathbb{SO}(3)$ is the orientation of frame $\{\mathcal{B}\}$ with respect to frame $\{\mathcal{I}\}$, $p \in \mathbb{R}^3$ and $v \in \mathbb{R}^3$ denote the position and the linear velocity of the rigid body expressed in the inertial frame $\{\mathcal{I}\}$, $p_i \in \mathbb{R}^3$ is the position of the $i$-th landmark expressed in $\{\mathcal{I}\}$, $g\in \mathbb{R}^3$ is the acceleration due to gravity expressed in $\{\mathcal{I}\}$, $a^\mathcal{B} \in \mathbb{R}^3$ is the apparent acceleration capturing all non-gravitational forces applied to the rigid body expressed in $\{\mathcal{B}\}$, and $\omega^\mathcal{B}$ is the angular velocity of the rigid body expressed in $\{\mathcal{B}\}$. Figure \ref{VIO_diagram} illustrates the VIO configuration for an example with four landmarks.

\begin{figure}[h]
    \centering
\includegraphics[width=0.99\linewidth]{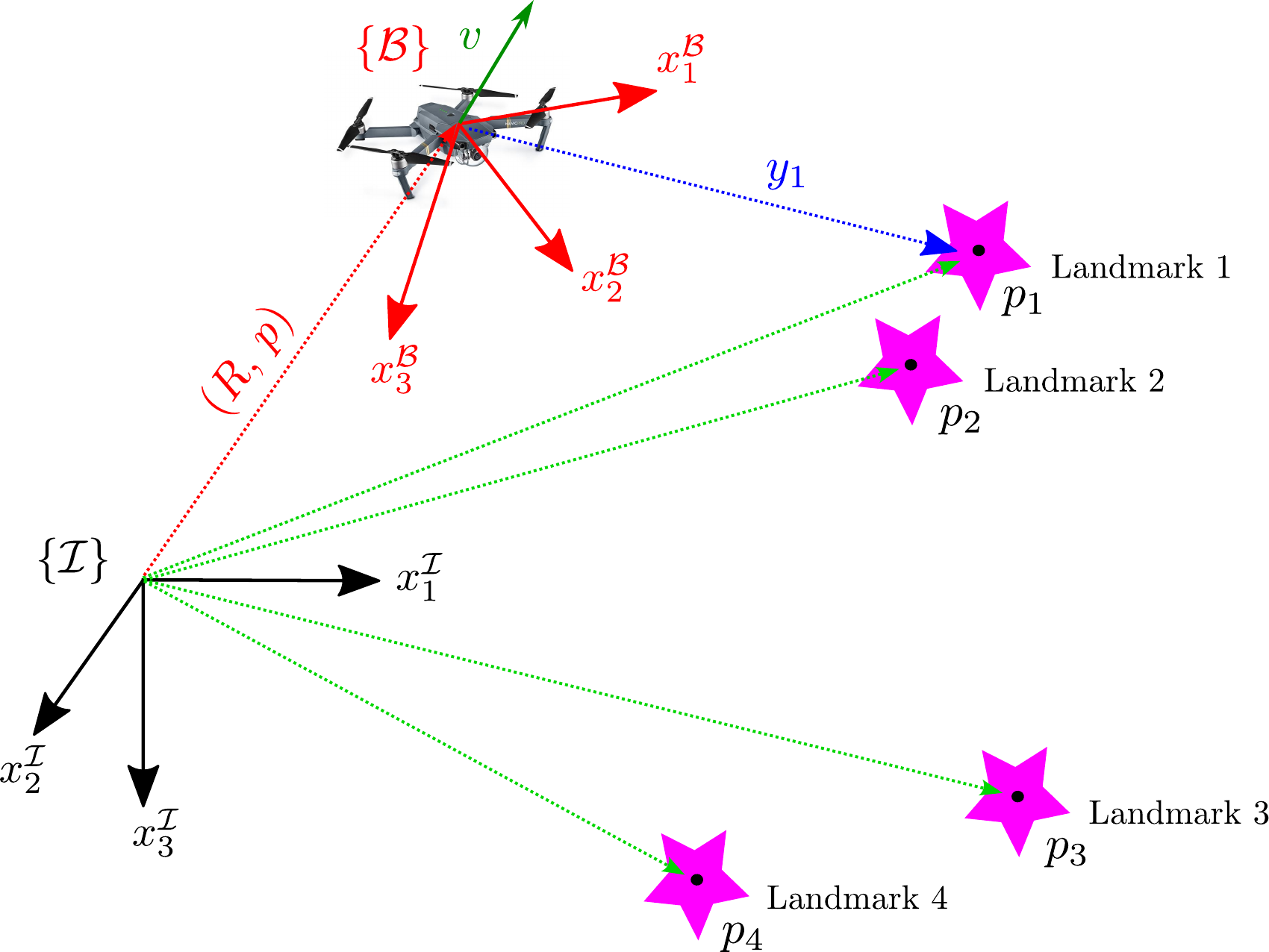}
    \caption{Rigid body navigation in 3D with four unknown landmarks. The goal is to estimate the vehicle’s extended pose (position, velocity, and orientation) based on landmark position estimates and available measurements.}
    \label{VIO_diagram}
\end{figure}

\subsection{System representation on a matrix Lie group} 

The system dynamics \eqref{equ:dynamics1}-\eqref{equ:dynamics4} can be captured by a state evolving on the Lie group $\mathbb{SE}_{2+n}(3)$ since the gravity vector $g$ is constant and known. However, it is difficult to design an observer with global stability guarantees directly on this group due to the coupling resulting from the gravity direction when considering the group error. To remove this coupling, we extend the system with the additional (auxiliary) state $g$, which satisfies $\dot g=0$. Although redundant from a physical standpoint, this auxiliary state eliminates the undesired coupling in the group error dynamics, thereby simplifying the observer design. The idea of introducing auxiliary states into the system to remove coupling in the error dynamics has, for example, been used in \cite{Wang_TAC2021} in the context of visual-inertial navigation with a known map. The resulting dynamics of \eqref{equ:dynamics1}-\eqref{equ:dynamics4}, augmented with $\dot g=0$, can be captured by considering the extended state $X=\mathcal{M}(R,p,v,g,\bold{p_{L}})$, with $\bold{p_L}:=[p_1~p_2~\hdots~p_n]$, which evolves on the higher-dimensional Lie group $\mathbb{SE}_{3+n}(3)$. Consequently, the system dynamics can be written compactly as
\begin{equation}\label{real_sys_on_group}
    \dot{X}=[X, H]+XV,
\end{equation}
where the group velocity $V=\mathcal{V}\left([\omega^\mathcal{B}]_\times, 0_{3\times 1}, a^\mathcal{B}, 0_{3\times 1}, 0_{3\times n}\right)$ is derived directly from IMU measurements. The operator $[.,.]$ denotes the matrix Lie bracket, defined as $[X_1,X_2]=X_1X_2-X_2X_1$ for any $X_1, X_2 \in \mathbb{R}^{(6+n)\times(6+n)}$. The constant matrix \(H\) is defined as
\begin{equation*} H=\left[\begin{array}{r|r} \begin{matrix} 0_3 \qquad \end{matrix} & \begin{matrix} 0_{3\times (3+n)}\end{matrix}\\ \hline 0_{(3+n)\times 3}& \qquad S \qquad \end{array}\right] \end{equation*}
where ~~~$S=\begin{bmatrix}
        0&0&0&\cdots&0&0\\
        1&0&0&\cdots&0&0\\
        0&1&0&\cdots&0&0\\
        0&0&0&\cdots&0&0\\
        \vdots&\vdots&\vdots&\ddots&\vdots&\vdots\\
          0&0&0&\cdots&0&0\\
    \end{bmatrix} \in \mathbb{R}^{(3+n)\times (3+n)}$.\\
    
\begin{rmk}
    In VIO, the IMU provides reliable short-term motion prediction, whereas visual measurements are essential for correcting accumulated drift. Within the group representation framework, this complementary relationship becomes explicit: the IMU determines the group velocity $V$, while visual observations constrain the system with respect to static landmarks.
\end{rmk}

\subsection{Measurement models}

We consider three visual measurement models: relative position, stereo bearing, and monocular bearing. The bearing models share a common formulation, with one camera for monocular sensing and two for stereo sensing. Each model is expressed using the matrix Lie group representation introduced above.

\subsubsection{Relative position measurements}

In this case, we assume that the relative position between the landmarks and the rigid body is measured in the body-fixed frame $\{\mathcal{B}\}$ and given as follows:
\begin{equation}\label{equ:measurements_3d}
    y_i=R^\top(p_i-p),
\end{equation}
where $i\in\{1,2,\ldots,n\}$.
Defining the vectors $\bold y_i:= [y_i^\top~b_i^\top]^\top \in\mathbb{R}^{6+n}$ and $\bold b_i:= [0_{3\times1}^\top~b_i^\top]^\top \in\mathbb{R}^{6+n}$,
where $b_i:= [1~0~0~-e_i^\top]^\top \in\mathbb{R}^{3+n}$ and $e_i$ denotes the $i$-th basis vector of $\mathbb{R}^n$, one can rewrite \eqref{equ:measurements_3d} as
\begin{equation}\label{equ:measurements_3d_on_group}
    \bold y_i=X^{-1}\bold b_i.
\end{equation}
This model describes measurements from sensors that provide metric relative positions, such as RGB-D cameras, and therefore directly supply landmark depth.

\subsubsection{Bearing measurements}

Let $N_c\in\{1,2\}$ denote the number of cameras. For each $q\in\{1,\ldots,N_c\}$, the known camera extrinsics $(R_{cq},p_{cq})$ specify the orientation and position of camera frame
$\{\mathcal C_q\}$ relative to the body frame $\{\mathcal B\}$. The bearing of landmark $i$, expressed in camera frame $\{\mathcal C_q\}$, is
\begin{equation}\label{equ:measurements_sb}
    \bar y_i^q=
    \frac{
        R_{cq}^\top
        \bigl(R^\top(p_i-p)-p_{cq}\bigr)
    }{
        \bigl\|R^\top(p_i-p)-p_{cq}\bigr\|
    }.
\end{equation}
Define $\bar{\bold y}_i^q:= [(\bar y_i^q)^\top~0_{1\times(3+n)}]^\top \in\mathbb R^{6+n}$, $X_{cq}:= \mathcal M(R_{cq},0_{3\times1},0_{3\times1}, 0_{3\times1},0_{3\times n})$,
and $\bold p_{cq,i}:=[p_{cq}^\top~b_i^\top]^\top$. The corresponding group representation is
\begin{equation}\label{equ:measurements_sb_on_group}
    \bar{\bold y}_i^q=
    \frac{
        X_{cq}^{-1}
        \bigl(X^{-1}\bold b_i-\bold p_{cq,i}\bigr)
    }{
        \bigl\|
        X_{cq}^{-1}
        \bigl(X^{-1}\bold b_i-\bold p_{cq,i}\bigr)
        \bigr\|
    }.
\end{equation}
Stereo sensing corresponds to $N_c=2$, with bearings from both cameras. For monocular sensing, $N_c=1$, and we omit the camera index by setting $(R_c,p_c):=(R_{c1},p_{c1})$ and
\begin{equation}\label{equ:measurements_mb}
    \bar y_i:=\bar y_i^1.
\end{equation}
Similarly, with $X_c:=X_{c1}$ and $\bold p_{c\,i}:=\bold p_{c1,i}$, the monocular group measurement is
\begin{equation}\label{equ:measurements_mb_on_group}
   \bar{\bold y}_i:=\bar{\bold y}_i^1,
\end{equation}
where the right-hand side is given by \eqref{equ:measurements_sb_on_group} with $q=1$.
Stereo bearings permit depth recovery through triangulation when the viewing geometry is nondegenerate. A single monocular bearing supplies only direction, so depth recovery requires
additional information over time, such as motion-induced parallax. The corresponding observability conditions are established in the subsequent analysis.

\begin{rmk}
For a calibrated pinhole camera, the bearing is obtained from undistorted pixel coordinates $(u_i,v_i)$ as
\begin{equation}
    \bar y_i=
    \frac{\mathcal K^{-1}z_i}
         {\|\mathcal K^{-1}z_i\|},
    \qquad
    z_i=[u_i,v_i,1]^\top,
\end{equation}
where $\mathcal K$ is the camera intrinsic matrix. The same construction applies separately to each stereo camera.
\end{rmk}

\begin{rmk}
The bearing measurements are well defined provided that no observed landmark coincides with the corresponding camera center, that is, $p_i\neq p+Rp_{cq}$ for each observed landmark--camera pair.
\end{rmk}

\subsection{Objective}
The objective is to design an observer that estimates the rigid body’s extended states (orientation, position, velocity, and gravity) together with the landmark positions, based on the dynamics \eqref{real_sys_on_group} and one of the measurement models \eqref{equ:measurements_3d_on_group}, \eqref{equ:measurements_sb_on_group}, or \eqref{equ:measurements_mb_on_group}. This corresponds to a VIO problem, where the goal is to estimate the rigid body pose relative to the landmarks while simultaneously estimating the landmark positions. It is well established that VIO is subject to fundamental observability limitations: the global position cannot be recovered (only relative motion is observable), and the orientation about the gravity direction remains unobservable.

\begin{rmk}
    The observability limitations are not consequences of the proposed formulation, but inherent to VIO itself. Any consistent estimator must respect these structural constraints. Attempting to estimate unobservable quantities typically leads to inconsistency and divergence in practical implementations.
\end{rmk}

\begin{rmk}
    Unlike SLAM, the objective of VIO is not to reconstruct a global map, but rather to maintain a consistent trajectory estimate. In this formulation, landmarks primarily serve as temporary anchors that support inertial odometry and can be discarded or marginalized as the system evolves.\\
\end{rmk}
 
 In the sequel, we propose an almost globally asymptotically stable observer that estimates the rigid body’s extended states (orientation, position, velocity, and gravity) together with the landmark positions, up to an unknown constant translation and a rotation about the $z$-axis. This ensures that the proposed scheme remains consistent, as it preserves the fundamental observability properties of the VIO problem.

\section{Observer Design}\label{s4}
In view of \eqref{real_sys_on_group}, we propose the following VIO observer structure on $\mathbb{SE}_{3+n}(3)$ with a copy of the dynamics plus an innovation term:
\begin{equation} \label{observer_sys_on_group}
    \dot{\hat X}=[\hat X, H]+\hat X V+\Delta \hat X,
\end{equation}
where $\hat X=\mathcal{M}(\hat R,\hat p,\hat v,\hat g,\bold{\hat p}_{\bold L})$ is the estimated state, with $\bold{\hat p}_{\bold L}:=[\hat p_1~\hat p_2 \hdots \hat p_n]$, and  $\Delta \in \mathfrak{se}_{3+n}(3)$ is an innovation term that will be designated later. Note that $\hat R \in \mathbb{SO}(3)$, $\hat p \in \mathbb{R}^3$, $\hat v \in \mathbb{R}^3$, $\hat g \in \mathbb{R}^3$ and $\hat p_i \in \mathbb{R}^3$ denote the estimates of $R$, $p$, $v$, $g$ and $p_i$, for each $i\in\{1, 2, \hdots, n\}$, respectively. The multiplicative correction structure $\Delta\hat X$ follows the invariant observer
framework of
\cite{barrau2016invariant}. In our design,
gravity is included as an auxiliary state,
and $\Delta$ is chosen to decouple the
translational error dynamics from the attitude
error, enabling the cascaded stability
analysis developed below. Next, we derive the innovation term for each landmark measurement obtained from \eqref{equ:measurements_3d_on_group}, \eqref{equ:measurements_sb_on_group}, or \eqref{equ:measurements_mb_on_group}. For each $i \in \{1,\, 2, \,\cdots,\, n\}$, we define the following innovation term corresponding to the $i$-th landmark measurement:
%\textcolor{red}{one can verify that (verify what??!)} 
\begin{equation}
    \bold{\tilde y}_i =
    \begin{cases}
        \hat X^{-1}\bold b_i - \bold y_i, 
        & \text{(relative-position)}, \\[4pt]
        \displaystyle\sum_{q=1}^2\pi(X_{cq}\bar{\bold y}_i^q)
        (\hat X^{-1}\bold b_i - \bold p_{cq,i}),
        & \text{(stereo bearing)}, \\[6pt]
        \pi(X_c\bar{\bold y}_i)(\hat X^{-1}\bold b_i - \bold p_{c,i}),
        & \text{(monocular bearing)}.
    \end{cases}\nonumber
\end{equation}
Moreover, by exploiting the properties of the projection map $\pi(\cdot)$, the innovation term associated with each landmark measurement can be written in the following form
\begin{equation}\label{meas_innv}
    \bold{\tilde y}_i = \Xi_i \left(\hat X^{-1}\,\bold b_i-X^{-1}\,\bold b_i\right),
\end{equation}
where
\begin{equation}\label{Xi_i_equ}
    \Xi_i =
    \begin{cases}
        I_{6+n}, 
        & \text{(relative-position)}, \\[4pt]
        \displaystyle\sum_{q=1}^{2}\pi\!\left(X_{cq}\bar{\bold y}_i^{q}\right),
        & \text{(stereo bearing)}, \\[6pt]
        \pi\!\left(X_{c}\bar{\bold y}_i\right),
        & \text{(monocular bearing)}.
    \end{cases}
\end{equation}
The projection map $\pi(\cdot)$, appearing in the definition of $\bold{\tilde y}_i$, is widely used for processing bearing measurements obtained from cameras \cite{HAMEL2017137,Wang_TAC2021}. This projection is crucial because it preserves directional information while discarding unknown depth, allowing the observer to exploit geometric constraints for correction even when absolute distances are not available.

Let $Q := \left[I_3~~~0_{3\times(3+n)}\right]$. We introduce the following innovation term, constructed directly from the landmark-measurement innovation terms given in \eqref{meas_innv}:
\begin{equation}
    \sigma^p := \operatorname{vec}\left(Q\,[\bold{\tilde y}_1,\,\bold{\tilde y}_2,\, \cdots,\, \bold{\tilde y}_n]\right).
\end{equation}
In pursuit of the previously defined objective, we propose the following innovation term to be incorporated into the VIO observer dynamics \eqref{observer_sys_on_group}:
\begin{align}\label{gen_innov}
    \Delta=&\mathcal{V}(k_R[\sigma^R]_\times,\bold K_p(t)\,\sigma^p, \bold 
 K_v(t) \,\sigma^p, \bold K_g(t)\,\sigma^p, (\bold \Gamma(t) \,\sigma^p)^\wedge ),
\end{align}
where $\sigma^R \in \mathbb{R}^3$, $k_R \in \mathbb{R}$, \(\mathbf K_*(t)\in\mathbb{R}^{3\times 3n}\), and \(\boldsymbol\Gamma(t)\in\mathbb{R}^{3n\times 3n}\). The proposed innovation term $\Delta$ consists of two main correction components: $\sigma^R$, which corrects the rotational motion (its expression will be designed later), and $\sigma^p$, which incorporates the landmark-based residuals required for translational correction. It is also important to note that, although $\Delta$ lies in $\mathfrak{se}_{3+n}(3)$, its construction deviates from gradient-based approaches relying on potential functions typically used on matrix Lie groups. 
%While such classical designs are common, they often lack strong stability guarantees \cite{Bonnabel_IEKFSLAM,Pieter_TR2023}({\color{red} I don't think that these references use a gradient based correcting term}). In contrast, the proposed method exhibits stronger stability properties, as will be rigorously demonstrated in the subsequent analysis.
The key idea behind designing $\Delta$ using two separate correction terms, $\sigma^R$ and $\sigma^p$, is to decouple the rotational motion from the translational motion, thereby allowing them to be treated as two cascaded subsystems. To see this, let us first derive the closed-loop system. Using the identity $\frac{d}{dt}(\hat X^{-1})=-\hat X^{-1} \dot{\hat X} \hat X^{-1}$ and defining the right-invariant geometric error $E=X \hat X^{-1}$, it follows from \eqref{real_sys_on_group} and \eqref{observer_sys_on_group} that
\begin{equation}\label{E_on_group}
    \dot E = [E, H]-E\Delta.
\end{equation}
Here $E=\mathcal{M}(\tilde R, \tilde p, \tilde v, \tilde g, \tilde{\bold{p}}_L)$ where $\tilde R:=R \hat R^T$, $\tilde p:=p-\tilde R \hat p$, $\tilde v:=v-\tilde R \hat v$, $\tilde g:=g-\tilde R \hat g$ and $\tilde{\bold{p}}_L:=[\tilde p_1~\tilde p_2~\hdots~\tilde p_n]$ with $\tilde p_i:=p_i-\tilde R \hat p_i$. 
Now, let the error state of the translational subsystem be defined as
\begin{equation}
    x=\left[(R^\top \tilde v)^\top, (R^\top \tilde g)^\top, \left(\operatorname{vec}(R^\top(\mathbf 1_{1\times n}\otimes \tilde p-\tilde{\mathbf{p}}_L))\right)^\top\right]^\top. \nonumber
\end{equation}
The choice of the above translational error state $x$ is motivated by the observability properties of the VIO system. The available measurements do not render absolute position or yaw observable. Instead, only body-frame velocity and gravity errors, together with relative landmark positions expressed in the body frame, can be reconstructed. Therefore, we define the translational error state using these observable combinations, which lie in the observable subspace of the VIO system and are invariant to the unobservable global position and yaw direction. Moreover, the correction term associated with the translational motion can be rewritten as
\begin{align}
    \sigma^p = \bold C(t) \,x,
\end{align}
where 
\begin{align}
    \bold C(t)&=\begin{bmatrix}
        0_{3}&0_3&\Pi_1&0_3& \cdots&0_3\\
        0_{3}&0_3&0_3&\Pi_2& \cdots&0_3\\
        \vdots&\vdots&\vdots&\vdots&\ddots&\vdots\\
        0_{3}&0_3&0_3&0_3& \cdots&\Pi_n
    \end{bmatrix} \label{matrix:C}
\end{align}
with $\Pi_i:=Q\,\Xi_i\,Q^\top$ and $\Xi_{i}$ defined in \eqref{Xi_i_equ}.
Accordingly, we derive the following closed-loop dynamics for the translational subsystem: 
\begin{equation}\label{ltv_cls}
    \dot x = \left(\bold A(t)-\bold L(t) \bold C(t)\right)x,
\end{equation}
where
\begin{align}
    \bold A(t)&=\begin{bmatrix}
        \bold D(t)&0_{6\times 3n}\\
        B_n \otimes I_3&I_n \otimes(-[\omega^\mathcal{B}]_\times)\\
    \end{bmatrix}\label{matrix:A}\\
    \bold L(t)&=\left(I_{n+2}\otimes \hat R^\top\right)\,\begin{bmatrix}
        \bold K_v(t)\\
        \bold K_g(t)\\
        \mathbf 1_{n\times1} \otimes \bold K_p(t)- \bold \Gamma(t)
    \end{bmatrix} \label{matrix:L}
\end{align}
with 
\begin{align}
    \bold D(t)=\begin{bmatrix}
        -[\omega^\mathcal{B}]_\times&I_3\\
        0_{3\times3}&-[\omega^\mathcal{B}]_\times
    \end{bmatrix}~\text{and}~B_n=\begin{bmatrix}
       \mathbf 1_{n\times1} & 0_{n \times 1}
    \end{bmatrix}.\label{B_definition}
\end{align} \\ \\
Furthermore, by designing $\sigma^R := \hat{g} \times g$,  we derive the following dynamics for the rotational subsystem 
\begin{align}
    \dot{\tilde R}&=\tilde R\left[-k_R \left(\tilde R^\top g \times g\right)-\Upsilon(t)x\right]_\times,\label{R_tilde}
\end{align}
where $\Upsilon(t):=\left[0_3~k_R [g]_\times \hat R~0_3~\hdots~0_3\right] \in \mathbb{R}^{3 \times 3(n+2)}$. One checks that $||\Upsilon(t)||_F=\sqrt{2} k_R ||g||$. Note that the innovation term for the rotational subsystem $\sigma^R$ depends solely on a single vector, namely $g$ and its estimate $\hat g$. As a result, the attitude can only be estimated up to an arbitrary rotation about $g$ \cite{Mahony_TAC2008}. To facilitate the stability analysis of \eqref{R_tilde}, we therefore derive the following reduced attitude dynamics:
\begin{align}
        \dot{\breve g}&=k_R \left[ \breve g \times g \right]_\times \breve g-[\breve g]_\times \Upsilon(t) x, \label{g_tilde}
    \end{align}
    where $\breve g:= \tilde R^\top g$. The matrix $[\breve g]_\times \Upsilon(t)$ is bounded since $||\breve g||=||g||$ and $||\Upsilon(t)||_F=\sqrt{2} k_R ||g||$. Now, let us establish the stability properties of the proposed VIO observer \eqref{observer_sys_on_group}. Considering \eqref{ltv_cls} and \eqref{g_tilde}, one obtains the following overall closed-loop system: 
    \begin{align}
    \dot{\breve g}&=k_R \left[ \breve g \times g \right]_\times \breve g-[\breve g]_\times \Upsilon(t) x \label{csc_sys1}\\
    \dot x &= \left(\bold A(t)-\bold L(t) \bold C(t)\right)x.\label{csc_sys2}
\end{align}
The above closed-loop system can be seen as a cascaded nonlinear system evolving on $\mathbb{S}^2_g \times \mathbb{R}^{3\left(n+2\right)}$ where $\mathbb{S}_g^2:=\{u \in \mathbb{R}^3:~u^\top u=||g||^2\}$. Before establishing the stability properties of the closed-loop system \eqref{csc_sys1}–\eqref{csc_sys2}, we first introduce an instrumental lemma that will be used in the subsequent stability analysis.

\begin{lem}\label{lem_uo}
Consider the matrix $\bold A(t)$ defined in \eqref{matrix:A} and the matrix $\bold C(t)$ defined in \eqref{matrix:C}. The pair $(\bold A(t), \bold C(t))$ is uniformly observable if one of the following conditions holds:
\begin{itemize}
    \item Relative-position measurements \eqref{equ:measurements_3d} or stereo-bearing measurements \eqref{equ:measurements_sb} are available, with at least one landmark being observed.
    \item Monocular-bearing measurements \eqref{equ:measurements_mb} are available, with at least one landmark being observed, and for each monocular-bearing measurement there exist constants $\delta^*>0$ and $\mu^*>0$ such that
    \begin{align}
        \frac{1}{\delta^*}\int_{t}^{t+\delta^*} 
        \pi\!\left(R(\tau)\, R_c\, \bar y_i(\tau)\right) d\tau 
        \;\geq\; \mu^* I_3,\label{ineq_uo}
    \end{align}
    for all $t \geq 0$.
\end{itemize}
\end{lem}

\begin{proof}[Proof sketch]
Let $\bold T(t)=I_{n+2}\otimes R^\top(t)$ and $\bar{\bold A}=\bold A(t)-\bold S(t)$. Using $\dot{\bold T}=\bold S\bold T$ and $\bold T\bar{\bold A}=\bar{\bold A}\bold T$, the state transition matrix can be expressed as $$\bold\Phi(t,\tau) =\bold T(t)\exp\!\bigl(\bar{\bold A}(t-\tau)\bigr) \bold T^{-1}(\tau).$$
This expression allows the observability Gramian to be rewritten in terms of the constant matrix $\bar{\bold A}$, the matrix $\bar{\bold C}$, and the measurement-dependent matrix $\bar{\bold P}(t)$. For relative-position measurements, $\bar{\bold P}=I_{3n}$, and the full rank of the observability matrix associated with $(\bar{\bold A},\bar{\bold C})$ gives a uniformly positive-definite Gramian. For stereo-bearing measurements, the same conclusion follows from the uniform positive definiteness of $\bar{\bold P}$ under the stated conditions. For monocular-bearing measurements, condition~\eqref{ineq_uo}, together with the observability of $(\bar{\bold A},\bar{\bold C})$, yields the required Gramian bound through \cite[Lemma~2.7]{HAMEL2017137}. Since $\bold T(t)$ is orthogonal, these bounds establish uniform observability of $(\bold A(t),\bold C(t))$. The complete proof is provided in Appendix~\ref{app_uo}.
\end{proof}
The relative-position-based design represents the most favorable scenario, as the measurements directly provide the relative positions of landmarks. In this case, the innovation term is linear in the error state, with $\bold C(t)$ being constant since $\Pi_i = I_3$ for each $i\in\{1, 2, \ldots, n\}$. Consequently, the system’s observability is independent of the platform’s motion. The stereo configuration lies between the relative position and monocular cases. While each camera individually provides bearing measurements without scale, the baseline between the two cameras introduces parallax\footnote{Parallax is the apparent displacement of an object viewed from two different positions, used in stereo vision to calculate depth by measuring the shift in position between corresponding points in two camera images.}, which allows depth to be inferred and resolves scale ambiguity. The time-varying nature of $\bold C(t)$ reflects the changing projections of landmarks in both cameras, though uniform observability is still achieved without requiring persistent platform excitation. From a practical standpoint, stereo cameras serve as a lightweight alternative to depth sensors, offering sufficient geometric constraints for state estimation, though they have higher hardware complexity than monocular setups. Finally, the monocular case presents the greatest challenge, as bearing-only measurements cannot resolve depth. In this case, $\bold C(t)$ depends on the time-varying camera motion, and observability requires persistent excitation, as formalized by inequality~\eqref{ineq_uo}. System motion, such as translation orthogonal to the line of sight, generates parallax, which helps disambiguate scale and ensures sufficient information for estimation. This highlights a key trade-off: while monocular setups are lightweight and cost-effective, their ability to guarantee consistent estimation heavily depends on the richness of the motion.

With Lemma \ref{lem_uo} in place, we can now state the main result of this work.
\begin{thm}\label{thm:main}
    Assume that the matrices $Q(t)$ and \(V(t)\) are strictly positive definite. Let \(\bold L(t)=P(t) \bold C(t)^\top Q^{-1}(t)\) where \(P(t)\) is the solution of the CRE \eqref{CRE}. Then, for any \(k_R>0\) and provided that the pair $(\bold A(t), \bold C(t))$ is uniformly observable, the equilibrium $(\breve g, x)=(g, 0)$ of the system \eqref{csc_sys1}-\eqref{csc_sys2} is almost globally asymptotically stable\footnote{ The equilibrium point is said to be almost globally asymptotically stable if it is stable, and attractive from all initial conditions except a set of zero Lebesgue measure.}.
\end{thm}

\begin{proof}[Proof sketch]
We first analyze the translational subsystem~\eqref{csc_sys2} using $\mathcal V_P=x^\top P^{-1}x$. The uniform positive lower and upper bounds on $P$, together with the CRE~\eqref{CRE}, yield $$ \dot{\mathcal V}_P =-x^\top P^{-1}VP^{-1}x -x^\top\bold C^\top Q^{-1}\bold Cx \leq-\rho\mathcal V_P,$$ for some $\rho>0$. Consequently, the equilibrium $x=0$ is uniformly globally exponentially stable. Next, we analyze the reduced-attitude subsystem~\eqref{csc_sys1}. With $x=0$, its equilibria are $\breve g=g$ and $\breve g=-g$. The functions $\mathcal L_1=\frac12\|g-\breve g\|^2$ and $\mathcal L_2=\frac12\|g+\breve g\|^2$ satisfy $$\dot{\mathcal L}_1=-k_R\|g\times\breve g\|^2, \qquad \dot{\mathcal L}_2=k_R\|g\times\breve g\|^2.$$ These relations establish almost global asymptotic stability of the desired equilibrium and the repelling character of the undesired equilibrium. The proof then studies the reduced-attitude subsystem with input $x$, using the linearization around $\breve g=-g$ and an ultimate-boundedness estimate for $\mathcal L_1$ to establish its almost global ISS property. Combining this property with the uniform global exponential stability of the translational subsystem gives almost global asymptotic stability of the cascaded system. The complete proof is provided in Appendix~\ref{app_main_thr}.
\end{proof}

The closed-loop system \eqref{csc_sys1}-\eqref{csc_sys2} consists of two cascaded subsystems: the first evolves on a reduced attitude dynamics manifold, $\mathbb{S}^2_g$, a two-dimensional (scaled) sphere, while the second evolves in Euclidean space. Due to the topological obstruction of the 2-dimensional sphere \cite{Markdahl_SCL2017}, the almost globally asymptotically stable result in Theorem \ref{thm:main} represents the strongest stability achievable with continuous design. Furthermore, it is important to emphasize that the proposed design for the VIO observer \eqref{observer_sys_on_group} is deterministic in nature, and the theoretical results in Theorem \ref{thm:main} remain valid for any uniformly positive definite matrices \( Q \) and \( V \). In other words, the convergence and stability guarantees do not rely on any probabilistic noise model. 

\begin{rmk}\label{rmk:observer_gains}
    The observer matrix gains $\bold K_p(t)$, $\bold K_v(t)$, $\bold K_g(t)$, and $\bold \Gamma(t)$ are derived from the matrix $\bold L(t)$. Specifically, in view of \eqref{matrix:L}, one has
    \begin{align}
        &\bold K_v(t)=\begin{bmatrix}
        I_{3} & 0_{3\times3} & 0_{3\times3n}
        \end{bmatrix}\, \left(I_{n+2}\otimes \hat R\right)\,\bold L(t)\\
        &\bold K_g(t)=\begin{bmatrix}
        0_{3\times3} & I_{3} & 0_{3\times3n}
        \end{bmatrix}\, \left(I_{n+2}\otimes \hat R\right)\,\bold L(t)\\
        &\mathbf 1_{n\times 1} \otimes \bold K_p(t) - \bold \Gamma(t) =\nonumber\\
        &~~~~~~~~~~\begin{bmatrix}
        0_{3n\times3} & 0_{3n\times3} & I_{3n}
        \end{bmatrix}\, \left(I_{n+2}\otimes \hat R\right)\,\bold L(t).\label{gain_kp_t}
    \end{align}
    Note that the gains $\bold K_v(t)$ and $\bold K_g(t)$ are uniquely determined by the value of $(I_{n+2}\otimes \hat R)\,\bold L(t)$, whereas $\bold K_p(t)$ and $\bold \Gamma(t)$ must be chosen to satisfy equation \eqref{gain_kp_t}. Consequently, multiple valid choices exist for $\bold K_p(t)$ and $\bold \Gamma(t)$. This non-uniqueness arises because the VIO system operates relative to the robot’s own frame of reference (ego-centric). In essence, the selection of $\bold K_p(t)$ and $\bold \Gamma(t)$ affects the estimated states but does not influence the closed-loop dynamics of the translational subsystem \eqref{ltv_cls}. A straightforward choice is to set
    \begin{align}
    &\bold K_p(t) = 0_{3\times 3n}\nonumber\\ 
    &\bold \Gamma(t) = - 
    \begin{bmatrix}
    0_{3n\times3} & 0_{3n\times3} & I_{3n}
    \end{bmatrix} \, \left(I_{n+2}\otimes \hat R\right)\,\bold L(t).\nonumber
    \end{align}
    Alternatively, one may exploit the static-environment property and determine the pair $(\bold K_p(t), \bold \Gamma(t))$ that minimizes an optimization criterion similar to that proposed in \cite{VANGOOR_aut_2021}.\\
\end{rmk}

In contrast to state-of-the-art EKF-based VIO estimators, such as MSCKF \cite{Mourikis_ICRA2007}, and more recent equivariant filters \cite{Pieter_TR2023}, which rely on local linearization and therefore guarantee only local convergence, the proposed VIO observer \eqref{observer_sys_on_group} achieves almost global asymptotic stability by decoupling the rotational and translational subsystems. Specifically, by introducing the auxiliary dynamics $\dot g=0$, the translational subsystem becomes a standalone linear time-varying system independent of the attitude subsystem, for which a Riccati-based design yields a globally exponentially stable translational observer. The rotational subsystem is treated separately, where we develop an almost globally input-to-state stable attitude observer. This cascaded structure enables the estimation of the rigid body’s extended state (orientation, position, velocity, and gravity), together with landmark positions, up to the unobservable directions inherent to VIO (a global translation and a rotation about the $z$-axis), while ensuring almost global stability. The main trade-off of our approach lies in the use of a constant, and hence non-optimal (in the Kalman filter sense), gain for the attitude observer (\ie, the gain $k_R$), since the Riccati-based design applies only to the translational subsystem. Nonetheless, this trade-off is the cost of achieving a much stronger theoretical guarantee, providing stability results that go beyond the local convergence offered by EKF- and EqF-based VIO designs.

\subsection{Implementation notes} \label{sec:implementation}
To facilitate implementation, we express the proposed VIO observer \eqref{observer_sys_on_group} in componentwise form as follows:
\begin{align}
    \dot{\hat R}&=\hat R[\omega^\mathcal B+k_R\,\hat R^\top\sigma^R]_\times,\label{equ:obr_dynamics1_TV}\\
    \dot{\hat p}&=k_R\,[\sigma^R]_\times\hat p+\hat v+\sum_{j=1}^{n}K_j^p\sigma_j^p,\label{equ:obr_dynamics2_TV}\\
    \dot{\hat v}&=k_R\,[\sigma^R]_\times\hat v+\hat g+\hat R a^\mathcal B+\sum_{j=1}^{n}K_j^v\sigma_j^p,\label{equ:obr_dynamics3_TV}\\
    \dot{\hat g}&=k_R\,[\sigma^R]_\times\hat g+\sum_{j=1}^{n}K_j^g\sigma_j^p,\label{equ:obr_dynamics4_TV}\\
    \dot{\hat p}_i&=k_R\,[\sigma^R]_\times\hat p_i+\sum_{j=1}^{n}\Gamma_{ij}\sigma_j^p,\qquad i=1,\ldots,n.\label{equ:obr_dynamics5_TV}
\end{align}
where $K_j^*(t)$ is the $j$-th $3\times3$ block of $\mathbf K_*(t)$, for $*\in\{p,v,g\}$, and $\Gamma_{ij}(t)$ is the $(i,j)$-th $3\times3$ block of $\boldsymbol{\Gamma}(t)$. Innovation terms are
{\small
\begin{align}
    \sigma_j^p &=
    \begin{cases}
        \hat R^\top(\hat p_j-\hat p)-y_j,
        & \text{for \eqref{equ:measurements_3d}},\\[4pt]
        \displaystyle
        \sum_{q=1}^{2}
        \pi(R_{cq}\bar y_j^q)
        \bigl(\hat R^\top(\hat p_j-\hat p)-p_{cq}\bigr),
        & \text{for \eqref{equ:measurements_sb}},\\[6pt]
        \pi(R_c\bar y_j)
        \bigl(\hat R^\top(\hat p_j-\hat p)-p_c\bigr),
        & \text{for \eqref{equ:measurements_mb}},
    \end{cases}
    \nonumber\\
    \sigma^R&=\hat g\times g, \nonumber
\end{align}
}
for $j\in\{1,\ldots,n\}$. Figure~\ref{obs_str} illustrates the observer structure.
\begin{figure}[t]
    \centering
    \includegraphics[width=\linewidth]
{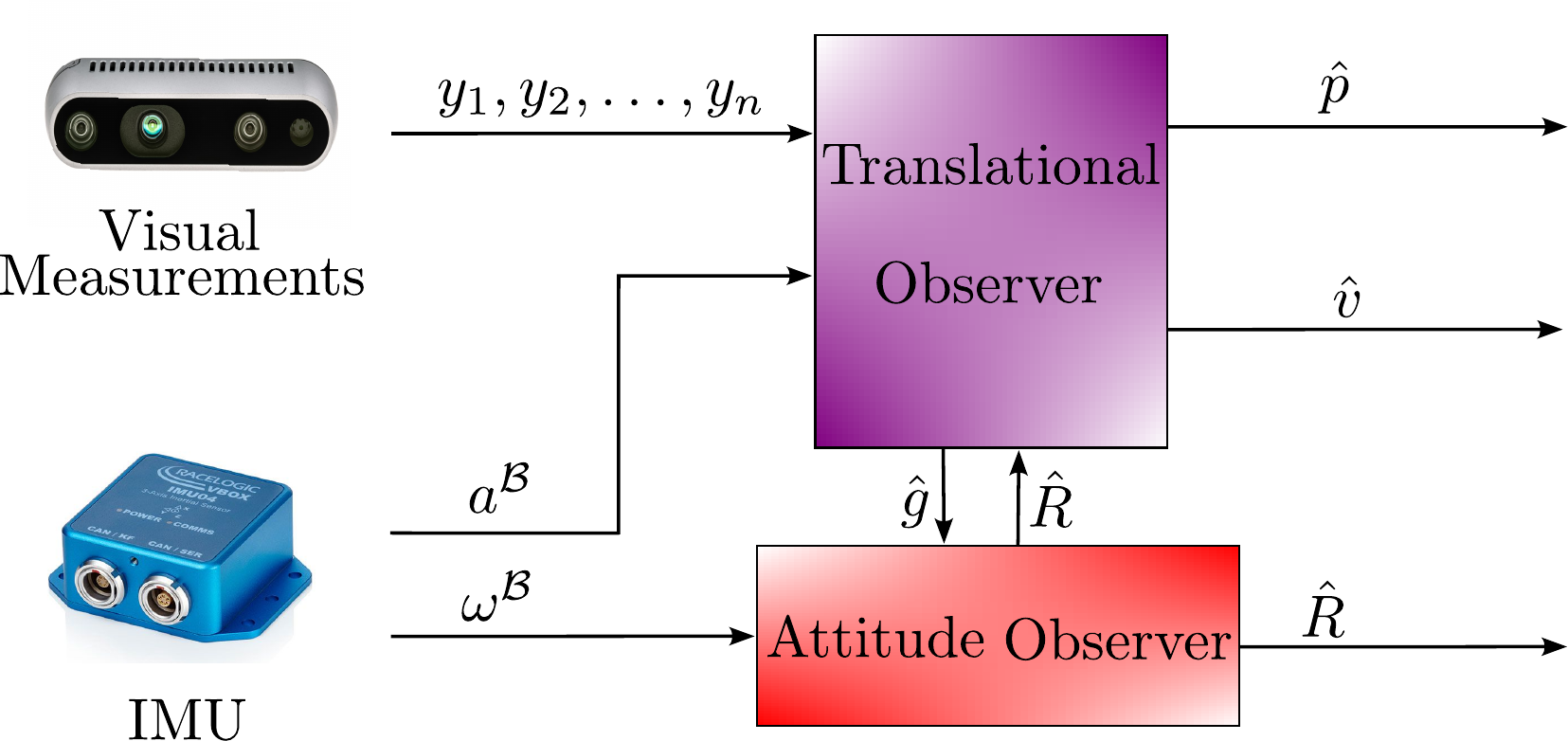}
    \caption{Structure of the proposed VIO observer.}
    \label{obs_str}
\end{figure}
In practical applications, IMU measurements are typically available at a higher rate than visual measurements. Motivated by the continuous--discrete framework in \cite{Wang_TAC2021}, we implement the proposed observer \eqref{observer_sys_on_group} through two steps: numerical propagation using IMU measurements between camera frames, and discrete correction when visual measurements become available. These propagation and correction steps are detailed below.

\subsubsection*{Propagation}
Between camera measurements, the visual correction
terms in \eqref{equ:obr_dynamics2_TV}--\eqref{equ:obr_dynamics5_TV} are omitted. Let $t_k$ denote the time of propagation sample $k$, and let $h_k=t_{k+1}-t_k$ be the corresponding
step duration. The IMU inputs are held constant over each step, and the attitude correction is evaluated as
\begin{equation}
    \sigma_k^R=\hat g_k\times g.
\end{equation}
The state estimates are propagated according to
\begin{align}
    \hat R_{k+1} &= \hat R_k \exp\!\left( h_k[\omega_k^{\mathcal B} +k_R\,\hat R_k^\top\sigma_k^R]_\times \right), \label{eq:implementation_attitude}\\
    \hat p_{k+1} &= \hat p_k+h_k\left( k_R\,[\sigma_k^R]_\times\hat p_k+\hat v_k \right), \label{eq:implementation_position}\\
    \hat v_{k+1} &= \hat v_k+h_k\left( k_R\,[\sigma_k^R]_\times\hat v_k +\hat g_k+\hat R_ka_k^{\mathcal B} \right), \label{eq:implementation_velocity}\\
    \hat g_{k+1} &= \hat g_k+h_k\,k_R\,[\sigma_k^R]_\times\hat g_k, \label{eq:implementation_gravity}\\
    \hat p_{i,k+1} &= \hat p_{i,k} +h_k\,k_R\,[\sigma_k^R]_\times\hat p_{i,k}, \qquad i=1,\ldots,n. \label{eq:implementation_landmarks}
\end{align}
where $\hat p_{i,k}$ denotes the estimate of landmark $i$ at time $t_k$. The attitude uses a matrix-exponential propagation, whereas the remaining states use forward Euler. All right-hand sides are evaluated using the same prior estimates. The resulting state propagation is first-order and preserves $\hat R\in \mathbb{SO}(3)$ in exact arithmetic. Upon the arrival of visual measurements, the propagated estimates serve as the prior estimates for the visual correction. The Riccati flow between visual corrections is
\begin{equation}
    \dot P=A(t)P+PA(t)^\top+V. \label{eq:implementation_riccati_flow}
\end{equation}
With the angular velocity and $V$ held constant over each propagation interval, the exact solution is
\begin{equation}
    P_{k+1} =\Phi_kP_k\Phi_k^\top+V_{d,k}, \label{eq:implementation_riccati}
\end{equation}
where
\begin{equation}
    \Phi_k=\exp(h_kA_k), \qquad V_{d,k} =\int_0^{h_k} e^{sA_k}Ve^{sA_k^\top}\,ds, \label{eq:implementation_process}
\end{equation}
and $A_k$ is the matrix $A(t)$ evaluated using the angular velocity $\omega_k^{\mathcal B}$. To exploit the structure of $A_k$, write
\begin{equation}
    A_k=N+D_k, \qquad D_k=I_{n+2}\otimes \bigl(-[\omega_k^{\mathcal B}]_\times\bigr),
\end{equation}
where $N=A_k-D_k$ is the constant coupling matrix. Since $ND_k=D_kN$ and $N^3=0$, it follows that
\begin{align}
    \Phi_k &= (I_{n+2}\otimes F_k) \left( I_{3n+6}+h_kN+\frac{h_k^2}{2}N^2 \right), \label{eq:implementation_structured_transition}
\end{align}
where $F_k=\exp\!\left(-h_k[\omega_k^{\mathcal B}]_\times\right)$. Thus, $\Phi_k$ can be evaluated using polynomial block operations and a $3\times3$ matrix exponential, without forming a full matrix exponential.
For
\begin{equation}
    V= \operatorname{diag} \left(q_vI_3,q_gI_3,q_\ell I_{3n}\right), \qquad q_v,q_g,q_\ell>0,
\end{equation}
the integral in \eqref{eq:implementation_process} is evaluated analytically using polynomial block expressions of degree at most five in $h_k$.
The propagated matrix is numerically symmetrized as $P\leftarrow(P+P^\top)/2$.

\subsubsection*{Visual correction}
At a propagation time $t_k$ coinciding with a camera measurement timestamp, let $n$ denote the number of maintained landmarks and let $m$ accepted visual measurements form the stacked innovation
$\sigma_k^p\in\mathbb{R}^{3m}$ and measurement matrix $C_k\in\mathbb{R}^{3m\times(3n+6)}$. Superscripts $-$ and $+$ denote values before and after correction. With measurement design matrix $Q_k=Q_k^\top\succ0$, define
\begin{align}
    S_k&=C_kP_k^-C_k^\top+Q_k,\\
    S_k\,\bold L_k^\top&=C_kP_k^-.
    \label{eq:implementation_gain}
\end{align}
The second equation is solved using an $LDL^\top$ factorization, without explicitly calculating $S_k^{-1}$.
Once $\bold L_k$ has been computed, the observer gains are recovered using the relations in Remark~\ref{rmk:observer_gains}, with $\hat R$ evaluated at the prior estimate $\hat R_k^-$. To select the position gain, we exploit the freedom to distribute the visual correction between the vehicle position estimate and the landmark position estimates. Specifically, we choose $\bold K_{p,k}$ to minimize the weighted norm of the landmark correction gains. Define
\begin{equation}
    Z_k=
    \begin{bmatrix}
        0_{3n\times3} & 0_{3n\times3} & I_{3n}
    \end{bmatrix}
    \left(I_{n+2}\otimes\hat R_k^-\right)
    \bold L_k.
\end{equation}
In view of \eqref{gain_kp_t}, $\bold\Gamma_k=\mathbf1_{n\times1}\otimes\bold K_{p,k}-Z_k$. The position gain is therefore selected through
\begin{equation}\label{K_P_opt}
    \underset{K\in\mathbb{R}^{3\times3m}}{\operatorname{minimize}}
    \quad
    \left\|
        W_k^{1/2}
        \left(\mathbf1_{n\times1}\otimes K-Z_k\right)
    \right\|_F^2,
\end{equation}
where $W_k= \operatorname{diag}(w_{1,k},\ldots,w_{n,k}) \otimes I_3$ with $w_{i,k}= \frac{1}{ \operatorname{tr}(P_{\ell_i\ell_i,k}^-)}$. Here, $P_{\ell_i\ell_i,k}^-$ denotes the $3\times3$ diagonal block of $P_k^-$ associated with the relative position estimate error of landmark $i$. One can verify that the unique minimizer of \eqref{K_P_opt} is
\begin{equation}
    \bold K_{p,k}^{\star} = \frac{ \left( [w_{1,k},\ldots,w_{n,k}]\otimes I_3 \right)Z_k }{ \sum_{i=1}^{n}w_{i,k} }.
\end{equation}
The implementation uses $\bold K_{p,k}=\alpha\bold K_{p,k}^{\star}$, where $\alpha\geq0$ controls the direct position correction, and $\bold\Gamma_k$ is then obtained from \eqref{gain_kp_t}. Setting $\alpha=1$ recovers the minimizer, whereas $\alpha=0$ gives the zero-position-gain choice described in Remark~\ref{rmk:observer_gains}. The inverse-trace weights and $\alpha$ are design choices. The state estimates are corrected according to
\begin{equation}
\begin{aligned}
    \hat R_k^+ &=\hat R_k^-,\\
    \hat p_k^+ &=\hat p_k^-+\bold K_{p,k}\sigma_k^p,\\
    \hat v_k^+ &=\hat v_k^-+\bold K_{v,k}\sigma_k^p,\\
    \hat g_k^+ &=\hat g_k^-+\bold K_{g,k}\sigma_k^p,\\
    \hat p_{i,k}^+ &=\hat p_{i,k}^-+\bold\Gamma_{i,k}\sigma_k^p, \qquad i=1,\ldots,n.
\end{aligned}
\label{eq:implementation_visual_update}
\end{equation}
where $\bold\Gamma_{i,k}$ is the $i$-th three-row block of $\bold\Gamma_k$. The Riccati matrix is corrected using the Joseph form:
\begin{align}
    M_k&=I_{3n+6}-\bold L_kC_k,\\
    P_k^+ &=M_kP_k^-M_k^\top +\bold L_kQ_k\bold L_k^\top. \label{eq:implementation_joseph}
\end{align}
The propagation and visual correction steps are summarized in Algorithm~\ref{alg:observer}.

\begin{rmk}
For the ideal continuous--discrete observer with noiseless measurements, a bias-free IMU, a fixed landmark set, and prescribed measurement selection and design matrices, the convergence analysis extends under suitable sampling conditions. Following \cite{Wang_TAC2021}, sufficient conditions include uniformly bounded positive-definite Riccati design matrices, uniform observability of the sampled system, and visual-correction times satisfying
\[
    0<T_{\min}\leq\tau_{j+1}-\tau_j
    \leq T_{\max}<\infty.
\]
These conditions ensure uniform Riccati bounds, yielding exponential convergence of the translational errors and almost global asymptotic convergence of the combined observable errors.
\end{rmk}
\begin{breakablealgorithm}
\caption{Nonlinear observer for VIO}
\label{alg:observer}
\hrule\vspace{2mm}
\begin{algorithmic}[1]

\StateWrap{\textbf{Input:} IMU samples and intermittent
visual measurements.}

\StateWrap{\textbf{Output:} Estimates
$\hat R,\hat p,\hat v,\hat g$, and $\hat p_i$.}

\StateWrap{Initialize the state estimates and $P$.}

\For{each visual frame}
    \While{propagation time $<$ frame timestamp}
        \StateWrap{Select the next propagation interval,
        ending at the frame timestamp if necessary.}

        \StateWrap{Compute
        $\sigma_k^R=\hat g_k\times g$.}

        \StateWrap{Propagate the state estimates using
        \eqref{eq:implementation_attitude}--%
        \eqref{eq:implementation_landmarks}.}

        \StateWrap{Propagate $P$ using
        \eqref{eq:implementation_riccati}.}
    \EndWhile

    \StateWrap{Remove landmarks no longer tracked and
    their corresponding rows and columns in $P$.}

    \StateWrap{Form $\sigma^p$, $C$, and $Q$ from accepted
    measurements of initialized landmarks.}

    \If{at least one measurement is accepted}
        \StateWrap{Compute $\bold L$ using
        \eqref{eq:implementation_gain}.}

        \StateWrap{Recover $\bold K_v$ and $\bold K_g$
        using Remark~\ref{rmk:observer_gains}.}

        \StateWrap{Compute $\bold K_p$ from the weighted
        minimizer of \eqref{K_P_opt}, scaled by $\alpha$.}

        \StateWrap{Obtain $\bold\Gamma$ from
        \eqref{gain_kp_t}.}

        \StateWrap{Correct the state estimates using
        \eqref{eq:implementation_visual_update}.}

        \StateWrap{Correct $P$ using
        \eqref{eq:implementation_joseph}.}
    \EndIf

    \StateWrap{Initialize admissible new landmarks
    and augment $P$.}
\EndFor

\end{algorithmic}
\end{breakablealgorithm}
The tuning parameters used in the implementation
are specified in Section~\ref{s7}.

\section{SIMULATION}\label{s6}

We evaluate the proposed observer through three simulation experiments: noiseless trials with large initial estimation errors, $50$ Monte Carlo trials with noisy IMU and visual measurements, and an additional $50$ Monte Carlo trials with visual outliers. The three visual measurement models \eqref{equ:measurements_3d}, \eqref{equ:measurements_sb}, and
\eqref{equ:measurements_mb} are considered and denoted by 3D, SB, and MB, respectively. An IEKF baseline adapted from \cite{Brossard_SJ2019} is included for comparison using 3D relative-position measurements.
The proposed 3D observer and the IEKF baseline receive identical IMU and visual measurements and use the same initial attitude, position, and velocity estimates.
The evaluation considers four estimation-error metrics: the norms of the body-frame velocity and gravity errors, the RMSE of body-frame relative landmark position errors, and the reduced-attitude error $\|\breve g-g\|$.

\subsubsection*{Trajectory and Landmark Setup}
The simulated vehicle follows a horizontal circular path of radius $3\,\mathrm{m}$ at a horizontal speed of $1\,\mathrm{m/s}$, together with a vertical
sinusoidal motion of amplitude $1.5\,\mathrm{m}$
and frequency $0.1\,\mathrm{Hz}$.
The position trajectory is
\begin{equation}
    p(t)=
    \begin{bmatrix}
        3\cos(t/3)\\
        3\sin(t/3)\\
        1.5\sin(0.2\pi t)
    \end{bmatrix},
    \qquad t\in[0,50]\,\mathrm{s}.
\end{equation}
The attitude is generated according to
\begin{equation}
    R(t)=R_z(\psi(t))R_y(\theta(t))R_x(\phi(t)),
\end{equation}
where $R_x$, $R_y$, and $R_z$ denote rotations about the corresponding coordinate axes, and $\psi(t)=t/3$, $\theta(t)=(3\pi/180)\sin(0.12\pi t)$, and $\phi(t)=(5\pi/180)\sin(0.16\pi t)$. Thus, roll and pitch oscillate with amplitudes of $5^\circ$ and $3^\circ$ at frequencies of $0.08\,\mathrm{Hz}$ and $0.06\,\mathrm{Hz}$, respectively. All angles in these equations are expressed in radians. The environment contains $1000$ landmarks randomly distributed on four vertical walls at $x=\pm6\,\mathrm{m}$ and $y=\pm6\,\mathrm{m}$, with heights in $[-5,5]\,\mathrm{m}$. At each camera frame, up to $50$ nearest visible landmarks are selected within a cone of total aperture $120^\circ$ about the body-frame $x$-axis, subject to positive camera-depth checks. Stereo measurements require visibility in both cameras. Landmarks leaving the selected set are removed, and newly selected landmarks are initialized from the available visual measurements. No loop closure or storage of a global map estimate is considered. Figure~\ref{fig:traj} illustrates the simulated trajectory and landmark environment.

\begin{figure}[t]
    \centering
    \includegraphics[width=\linewidth]
    {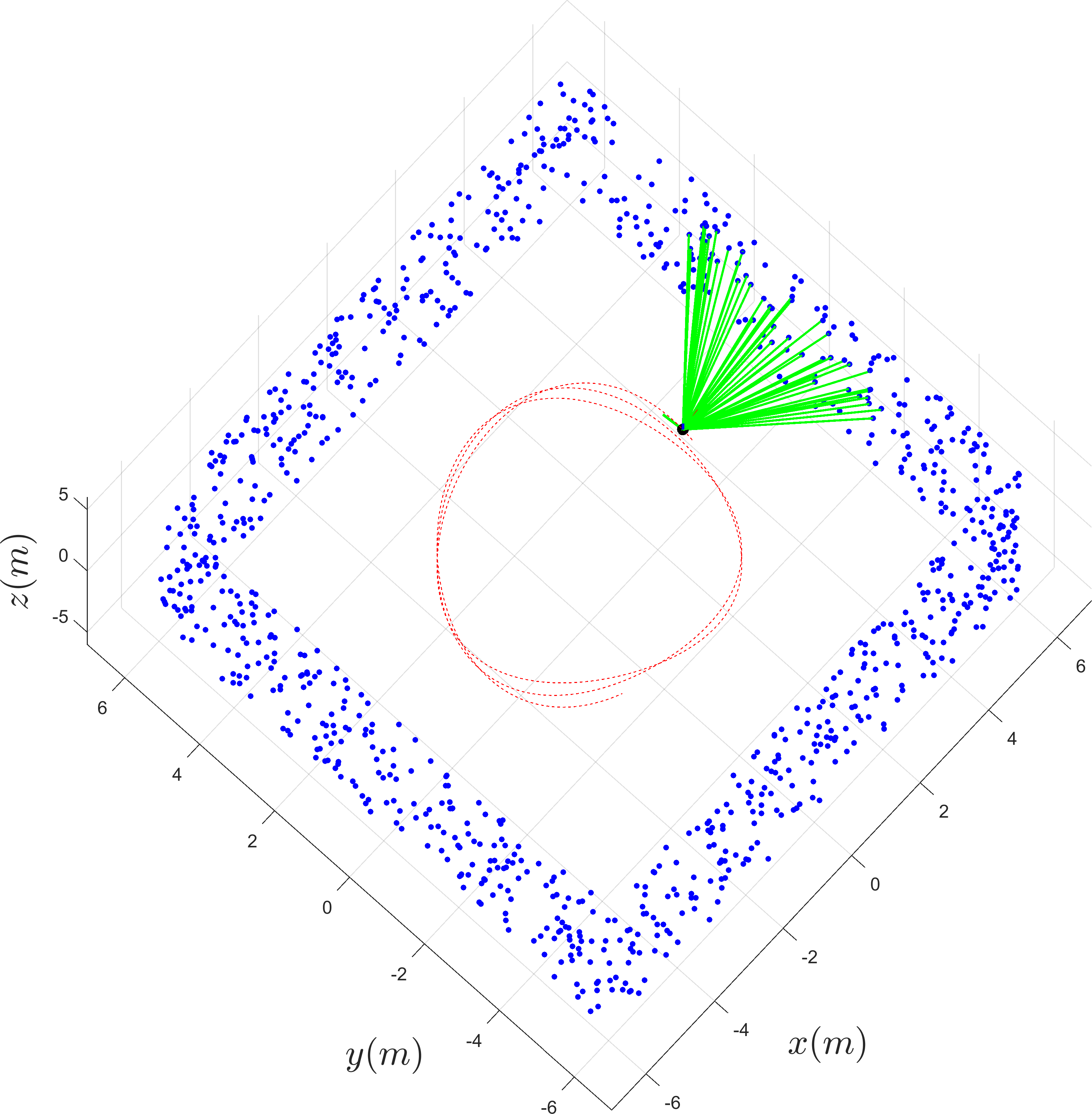}
    \caption{Simulated vehicle trajectory and landmark
    environment. An illustrative animation of the
    motion and environment is available at
    \url{https://youtu.be/ye6oyXOtPiM}.}
    \label{fig:traj}
\end{figure}

\subsubsection*{Sensor Models and Noise}
The IMU and cameras are sampled at $200\,\mathrm{Hz}$ and $20\,\mathrm{Hz}$, respectively. The camera-to-body extrinsic parameters are taken from the EuRoC MAV dataset \cite{EuRoC_dataset}. In the noisy experiments, independent zero-mean Gaussian noise is added to each IMU measurement axis, with sample standard deviations of $0.0024\,\mathrm{rad/s}$ for the gyroscope and $0.0283\,\mathrm{m/s^2}$ for the accelerometer. Bearing measurements are perturbed by small random rotations generated from two independent zero-mean Gaussian components, each with a standard deviation of $0.1^\circ$. These perturbations preserve unit-length bearings and are independent between the two cameras. The 3D relative-position measurements are corrupted by independent zero-mean Gaussian noise with a standard deviation of $0.05\,\mathrm{m}$ per axis. No IMU biases are included. All measurement noise is disabled in the noiseless experiments.

\subsubsection*{Initialization}
The inertial gravity vector is set to
$g=[0,\,0,\,9.81]^\top\,\mathrm{m/s^2}$. For the proposed observer, the state estimates are initialized as $\hat R(0)=\exp\!\left(\theta_0[u]_\times\right)R(0)$, $\hat p(0)=p(0)+[2,-1,1]^\top$, $\hat v(0)=v(0)+[1.5,-1,0.75]^\top$, and $\hat g(0)=g+[2,-1,1]^\top$, where $u$ is a horizontal unit vector and $\theta_0$ is the initial attitude-error angle. The position, velocity, and gravity offsets are expressed in $\mathrm{m}$, $\mathrm{m/s}$, and $\mathrm{m/s^2}$, respectively. The choices of $\theta_0$ and $u$ for each experiment are specified below. New landmarks are initialized using their corresponding visual measurements. The 3D configuration uses measured relative positions, whereas the stereo configuration uses triangulation. For monocular measurements, each landmark is initialized along its measured camera ray at an assumed range of $4\,\mathrm{m}$. The IEKF baseline uses the same initial attitude, position, and velocity estimates and the same 3D measurements as the proposed 3D observer.

\subsubsection*{Observer Parameters}

The attitude gain and position-correction scale are set to $k_R=0.1$ and $\alpha=0.6$, respectively. The process and measurement design matrices are selected as
\begin{equation}
    V=\operatorname{diag}\left(q_vI_3,q_gI_3,q_\ell I_{3n}\right),\qquad Q_k=q_yI_{3m},
\end{equation}
where $n$ is the number of active landmarks and $m$ is the number of visual measurements used in the correction.
The design weights are listed in
Table~\ref{tab:simulation_parameters} and remain unchanged between the noiseless and noisy experiments.

\begin{table}[t]
    \centering
    \caption{Observer design weights used in the simulations.}
    \label{tab:simulation_parameters}
    \begin{tabular}{lcccc}
        \toprule
        Model & $q_v$ & $q_g$ & $q_\ell$ & $q_y$ \\
        \midrule
        3D & $10^{-4}$ & $10^{-5}$ & $10^{-3}$
           & $2.5\times10^{-3}$ \\
        SB & $10^{-6}$ & $10^{-7}$ & $10^{-3}$
           & $0.05$ \\
        MB & $10^{-4}$ & $10^{-5}$ & $10^{-3}$
           & $0.05$ \\
        \bottomrule
    \end{tabular}
\end{table}

\subsubsection*{Convergence from Large Initial Errors}
Three noiseless trials are performed with $\theta_0\in\{100^\circ,150^\circ,175^\circ\}$ and $u=[1,0,0]^\top$, together with the initial position, velocity, and gravity offsets specified above. Figure~\ref{fig:simulation_noiseless} shows the resulting estimation errors using logarithmic vertical axes. Increasing line widths correspond to increasing initial attitude-error angles, i.e., thin, medium, and thick lines represent initial attitude errors of $100^\circ$, $150^\circ$, and $175^\circ$, respectively. For all three initializations and measurement configurations, the proposed observer substantially reduces the body-frame velocity, gravity, and reduced-attitude errors. The IEKF baseline also exhibits convergence from the initial errors in these trials. The relative-landmark RMSE remains nonzero, particularly in the monocular configuration, where its final value is approximately $0.28\,\mathrm{m}$. Unlike the fixed landmark set considered in the stability analysis, this simulation includes landmark removal and initialization as visibility changes. Newly initialized monocular landmarks use an assumed range and introduce fresh estimation errors into the active set. Consequently, the RMSE over this changing set need not converge to zero, even with noiseless measurements.

\begin{figure*}[t]
    \centering
    \includegraphics[width=\textwidth]
    {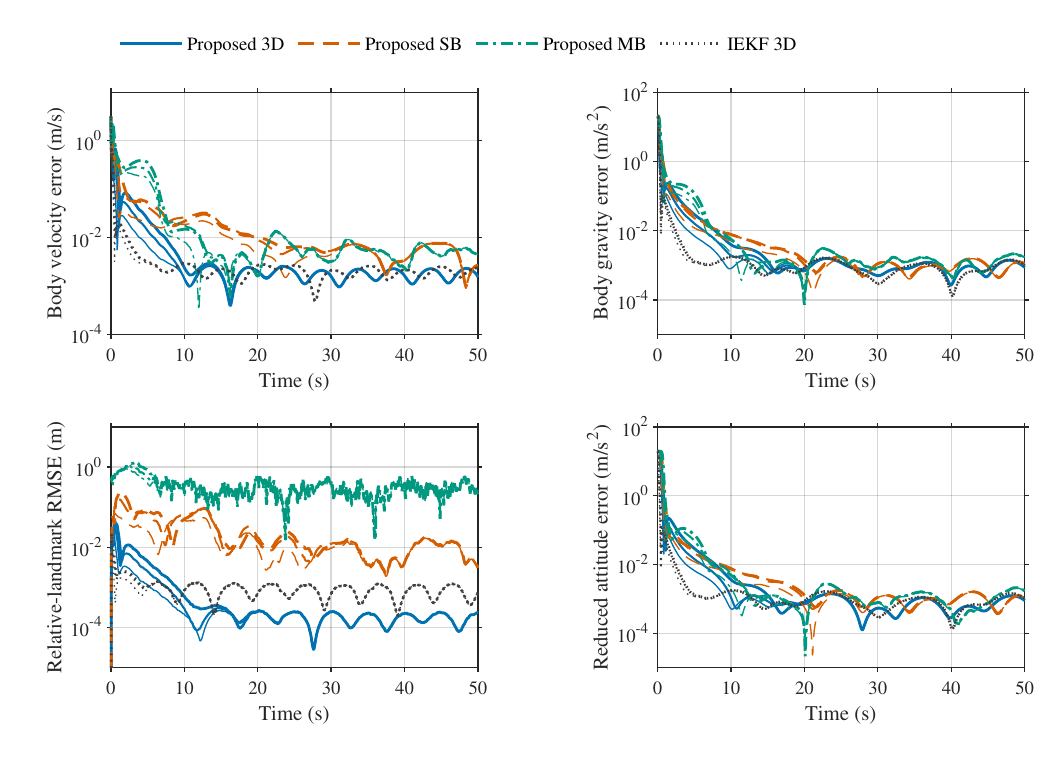}
    \caption{Estimation errors in the noiseless
    simulations. Increasing line widths correspond
    to initial attitude errors of $100^\circ$,
    $150^\circ$, and $175^\circ$, respectively.}
    \label{fig:simulation_noiseless}
\end{figure*}

\subsubsection*{Performance under Measurement Noise}
We perform $50$ Monte Carlo trials using the noise settings specified above. Each trial starts with an attitude error of $100^\circ$ about a randomly selected horizontal axis, together with the prescribed position, velocity, and gravity offsets. Figure~\ref{fig:simulation_noisy} presents the root-mean-square estimation errors across the $50$ trials, using logarithmic vertical axes. The errors decrease from their initial values and subsequently fluctuate at nonzero levels under measurement noise. Over the final $20\,\mathrm{s}$, the body-frame velocity RMS errors are $0.0088$, $0.0304$, and $0.0164\,\mathrm{m/s}$ for the proposed 3D, SB, and MB configurations, respectively. The IEKF baseline achieves $0.0064\,\mathrm{m/s}$ under the same 3D measurement model.

\begin{figure*}[t]
    \centering
    \includegraphics[width=\textwidth]
    {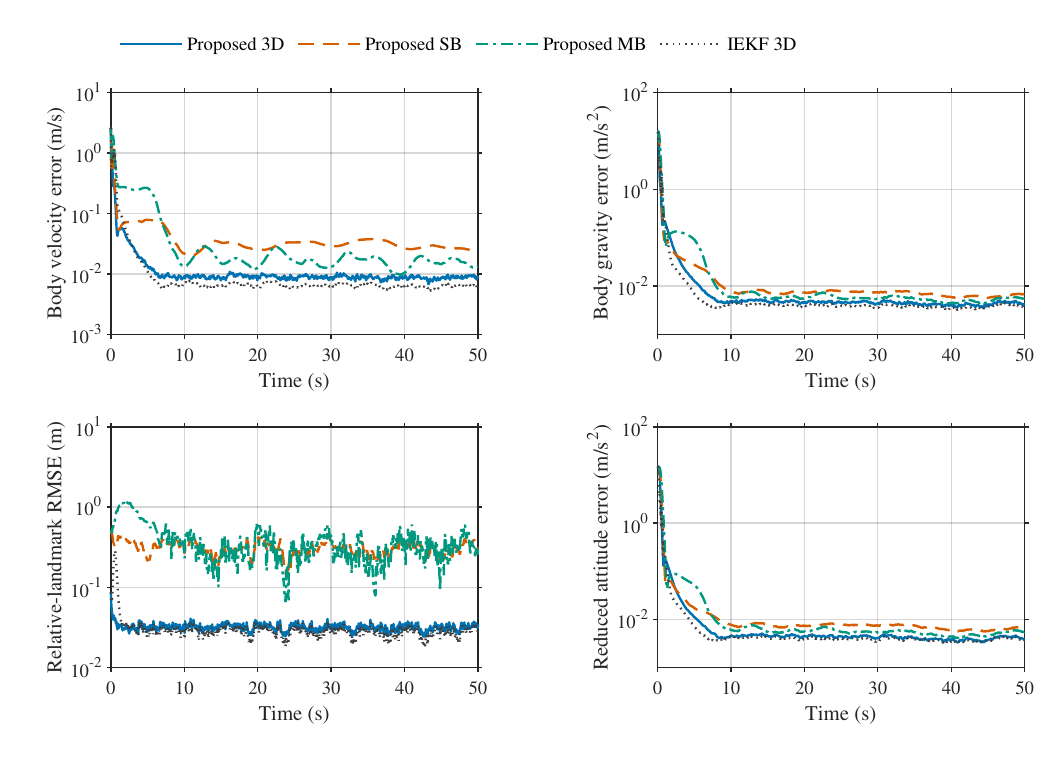}
    \caption{Root-mean-square estimation errors
    across $50$ Monte Carlo trials with noisy
    IMU and visual measurements.}
    \label{fig:simulation_noisy}
\end{figure*}

\subsubsection*{Sensitivity to Visual Outliers}

We conduct an additional $50$ Monte Carlo trials using the same initialization, ordinary measurement noise, and observer parameters as above. At each camera frame, each selected landmark has a $5\%$ probability of receiving an additional visual measurement perturbation. For 3D measurements, this perturbation is zero-mean Gaussian with a standard deviation of $1\,\mathrm{m}$ per axis. For bearing measurements, an additional random rotation is generated from two independent zero-mean Gaussian tangent components, each with a standard deviation of $10^\circ$. These perturbations are applied throughout the simulation and affect both landmark initialization and subsequent visual corrections.

Table~\ref{tab:simulation_outliers} compares the RMS errors computed across all $50$ trials and all logged samples over the final $20\,\mathrm{s}$, with and without the additional visual outliers.
\begin{table*}[t]
    \centering
    \small
    \setlength{\tabcolsep}{5pt}
    \caption{RMS errors over $t\in[30,50]\,\mathrm{s}$
    across $50$ Monte Carlo trials.
    ``Noise'' denotes ordinary measurement noise;
    ``Noise + outliers'' additionally includes
    the visual perturbations described in the text.
    Reduced-attitude error is measured by
    $\|\breve g-g\|$.}
    \label{tab:simulation_outliers}
    \begin{tabular}{llcccc}
        \toprule
        Method & Measurements
        & Body velocity
        & Body gravity
        & Relative landmarks
        & Reduced attitude \\
        & & (m/s) & (m/s$^2$) & (m) & (m/s$^2$) \\
        \midrule
        Proposed 3D & Noise
        & 0.0088 & 0.0043 & 0.0315 & 0.0041 \\
        & Noise + outliers
        & 0.0364 & 0.0099 & 0.1441 & 0.0074 \\
        \midrule
        Proposed SB & Noise
        & 0.0304 & 0.0067 & 0.3080 & 0.0067 \\
        & Noise + outliers
        & 0.2594 & 0.0262 & 1.0259 & 0.0244 \\
        \midrule
        Proposed MB & Noise
        & 0.0164 & 0.0053 & 0.3236 & 0.0052 \\
        & Noise + outliers
        & 0.9744 & 0.1650 & 2.7383 & 0.1346 \\
        \midrule
        IEKF 3D & Noise
        & 0.0064 & 0.0038 & 0.0280 & 0.0038 \\
        & Noise + outliers
        & 0.0197 & 0.0074 & 0.1273 & 0.0074 \\
        \bottomrule
    \end{tabular}
\end{table*}
All four methods complete the $50$ trials without recorded numerical failures. Nevertheless, the additional visual outliers increase the estimation errors, with the largest degradation observed in the monocular configuration. The IEKF baseline achieves smaller body-frame velocity errors than the proposed 3D observer under the same corrupted relative-position measurements.

\section{Implementation and Experimental Results}
\label{s7}
The proposed observer is implemented in C++ using ROS and evaluated on all eleven sequences of the \textit{EuRoC MAV} dataset~\cite{EuRoC_dataset}, covering the machine hall and both Vicon room environments. Both stereo and monocular configurations are considered. The evaluation examines position
accuracy, body-frame velocity and gravity estimation, and computational cost. Published results for stereo Basalt VIO, OKVIS, and VINS-Fusion, and monocular EqVIO, MSCKF, and VINS-Mono provide additional comparisons.

\subsubsection*{Experimental Setup}
Images and IMU measurements are provided at $20\,\mathrm{Hz}$ and $200\,\mathrm{Hz}$, respectively. The monocular configuration uses \texttt{cam0}, whereas the stereo configuration uses both cameras. The supplied camera calibration and camera--IMU extrinsics are used. Both configurations initialize position and velocity to zero and determine the initial attitude from the mean of the first 200 bias-compensated accelerometer samples, assuming an initially stationary platform. The auxiliary gravity estimate is initialized as $\hat g_0=[0,0,-9.81]^\top\,\mathrm{m/s^2}$. The gyroscope and accelerometer biases provided in the dataset's ground-truth files are subtracted from the corresponding IMU measurements throughout each sequence. These biases are neither included in the observer state nor estimated online. The proposed observer is implemented in C++ for both stereo and monocular configurations. Videos illustrating the experimental results are available for the stereo\footnote{\url{https://www.youtube.com/watch?v=yrUn6HIJ6gg}} and monocular\footnote{\url{https://www.youtube.com/watch?v=vlNjVwSzCt4&t=132s}} configurations.

\subsubsection*{Observer Parameters}
The process design matrix is chosen as $V=\operatorname{diag}(q_vI_3,q_gI_3,q_\ell I_{3n})$. Table~\ref{tab:observer_tuning} lists the attitude gain $k_R$, the position gain scale $\alpha$, and the process design weights. For each configuration, these parameters are kept fixed across all eleven sequences. For monocular measurements, each measurement design block is $Q_{i,k}=1.35\times10^{-5}I_3$. For stereo measurements, each block is positive definite and depends on the viewing geometry. 

\begin{table}[t]
    \centering
    \caption{Observer tuning parameters.}
    \label{tab:observer_tuning}
    \begin{tabular}{lcc}
        \toprule
        Parameter & Stereo & Monocular \\
        \midrule
        $k_R$    & $0.1$ & $0.3$ \\
        $\alpha$ & $0.6$ & $0$ \\
        $q_v$    & $0.35$ & $1.4\times10^{-7}$ \\
        $q_g$    & $3.0\times10^{-4}$
                 & $2.25\times10^{-4}$ \\
        $q_\ell$ & $1.5\times10^{-5}$
                 & $1.875\times10^{-5}$ \\
        \bottomrule
    \end{tabular}
\end{table}

\subsubsection*{Evaluation Protocol}
Position accuracy is evaluated using the root mean square error (RMSE) after aligning the estimated trajectory with the ground-truth trajectory \cite{Scaramuzza_IEEERSJ2018}. Table~\ref{tab:euroc_results} includes published stereo results for Basalt VIO, OKVIS, and VINS-Fusion from \cite{Usenko_RAL2020}, monocular EqVIO results from \cite{Pieter_TR2023}, and monocular MSCKF and VINS-Mono results from \cite{Scaramuzza_ICRA2018}. The stereo and monocular columns labeled VINS correspond to VINS-Fusion and VINS-Mono, respectively. The selected stereo results correspond to local VIO estimates without global mapping or loop-closure corrections. The selected VINS-Mono results also use the configuration without loop closure. The baseline values follow the evaluation protocols of their respective sources. Therefore, differences in initialization and IMU bias treatment should be considered when interpreting the reported errors.

\subsubsection*{Position Estimation Results}
Table~\ref{tab:euroc_results} reports the position RMSE of the proposed stereo and monocular observers, together with published results for stereo Basalt VIO, OKVIS, and VINS-Fusion, and monocular EqVIO, MSCKF, and VINS-Mono. Both configurations complete all eleven EuRoC sequences, including challenging recordings with rapid maneuvers, variations in illumination, and motion blur. This provides empirical evidence of robustness across the evaluated machine hall and Vicon room environments. The stereo observer achieves lower position errors than the monocular observer on eight sequences, with RMSE below $0.20\,\mathrm{m}$ on nine sequences. Its reported errors are lower than those of OKVIS on six of the ten shared sequences and VINS-Fusion on seven of these sequences. Compared with Basalt VIO, it achieves lower errors on \texttt{MH\_04} and \texttt{V1\_02}, and the same reported error on \texttt{MH\_02}. Larger stereo errors occur on \texttt{V2\_02} and \texttt{V2\_03}, showing that completing a sequence does not necessarily imply high position accuracy. The proposed monocular observer achieves lower reported errors than MSCKF on nine sequences, while EqVIO and VINS-Mono generally report lower errors. The results therefore demonstrate operation across all evaluated sequences, without consistent superiority in position accuracy over the published baselines. The stereo comparison in \cite{Usenko_RAL2020} excludes \texttt{V2\_03} because of missing frames from one camera. The proposed stereo observer completes this sequence, although its position RMSE reaches $1.50\,\mathrm{m}$.

\begin{table*}[t]
    \centering
    \small
    \setlength{\tabcolsep}{4pt}
    \caption{Position RMSE (m) on the EuRoC MAV dataset.
    The proposed observer uses trajectory alignment
    and IMU bias compensation based on the dataset's
    ground-truth files. Baseline values are taken
    from the cited publications under their respective
    evaluation protocols. A dash denotes a result
    not reported in the source.}
    
    \label{tab:euroc_results}
    \begin{tabular}{llcccccccc}
        \toprule
        & & \multicolumn{4}{c}{Stereo}
          & \multicolumn{4}{c}{Monocular} \\
        \cmidrule(lr){3-6}
        \cmidrule(lr){7-10}
        Sequence & Difficulty
        & Proposed
        & Basalt VIO \cite{Usenko_RAL2020}
        & OKVIS \cite{Usenko_RAL2020}
        & VINS \cite{Usenko_RAL2020}
        & Proposed
        & EqVIO \cite{Pieter_TR2023}
        & MSCKF \cite{Scaramuzza_ICRA2018}
        & VINS \cite{Scaramuzza_ICRA2018} \\
        \midrule
        MH\_01 & Easy
        & 0.17 & 0.07 & 0.23 & 0.24
        & 0.24 & 0.14 & 0.42 & 0.27 \\
        MH\_02 & Easy
        & 0.06 & 0.06 & 0.15 & 0.18
        & 0.16 & 0.20 & 0.45 & 0.12 \\
        MH\_03 & Medium
        & 0.09 & 0.07 & 0.23 & 0.23
        & 0.20 & 0.09 & 0.23 & 0.13 \\
        MH\_04 & Difficult
        & 0.09 & 0.13 & 0.32 & 0.39
        & 0.27 & 0.22 & 0.37 & 0.23 \\
        MH\_05 & Difficult
        & 0.15 & 0.11 & 0.36 & 0.19
        & 0.75 & 0.28 & 0.48 & 0.35 \\
        \midrule
        V1\_01 & Easy
        & 0.05 & 0.04 & 0.04 & 0.10
        & 0.22 & 0.06 & 0.34 & 0.07 \\
        V1\_02 & Medium
        & 0.04 & 0.05 & 0.08 & 0.10
        & 0.13 & 0.14 & 0.20 & 0.10 \\
        V1\_03 & Difficult
        & 0.15 & 0.10 & 0.13 & 0.11
        & 0.30 & 0.19 & 0.67 & 0.13 \\
        \midrule
        V2\_01 & Easy
        & 0.13 & 0.04 & 0.10 & 0.12
        & 0.09 & 0.11 & 0.10 & 0.08 \\
        V2\_02 & Medium
        & 0.32 & 0.05 & 0.17 & 0.10
        & 0.28 & 0.17 & 0.16 & 0.08 \\
        V2\_03 & Difficult
        & 1.50 & -- & -- & --
        & 0.38 & 0.20 & 1.13 & 0.21 \\
        \bottomrule
    \end{tabular}
\end{table*}

\subsubsection*{Body-Frame Velocity and Gravity}
Figures~\ref{fig:velocity_body} and \ref{fig:gravity_body} compare the estimated body-frame velocity $\hat R^\top\hat v$ and gravity $\hat R^\top\hat g$ with the corresponding reference values $R^\top v$ and $R^\top g$ on \texttt{MH\_04} and \texttt{V2\_03}. Both stereo and monocular configurations closely follow the reference on \texttt{MH\_04}. On \texttt{V2\_03}, the gravity estimates remain close to the reference, whereas larger deviations appear in the stereo velocity estimates, particularly in the second and third components. These deviations accompany the larger stereo position error reported in Table~\ref{tab:euroc_results}.

\begin{figure}[t]
    \centering
    \includegraphics[width=\columnwidth]
    {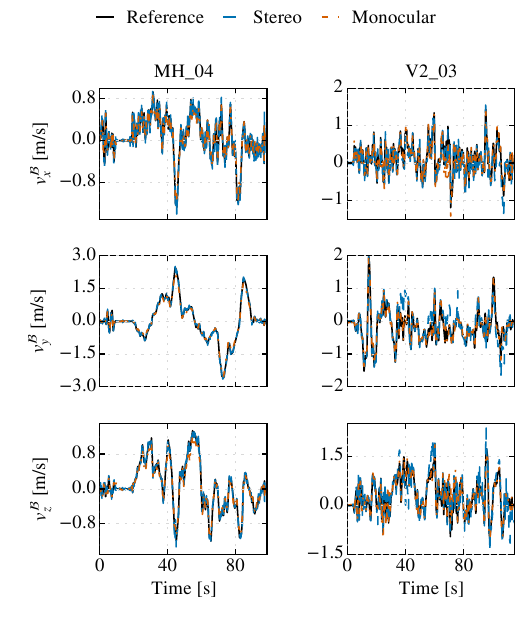}
    \caption{Body-frame velocity estimates compared
    with the reference on \texttt{MH\_04} (left)
    and \texttt{V2\_03} (right).}
    \label{fig:velocity_body}
\end{figure}

\begin{figure}[t]
    \centering
    \includegraphics[width=\columnwidth]
    {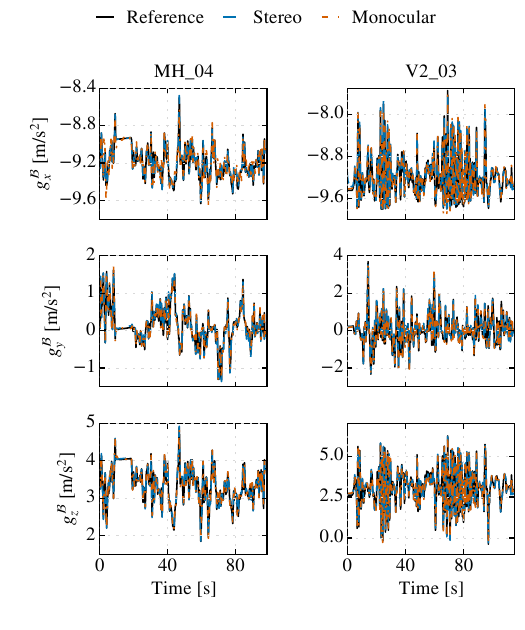}
    \caption{Body-frame gravity estimates compared
    with the reference on \texttt{MH\_04} (left)
    and \texttt{V2\_03} (right).}
    \label{fig:gravity_body}
\end{figure}

\subsubsection*{Computational Cost}
Computational cost was evaluated on an Intel Core i3-3120M processor (2.50\,GHz) using a Release build and playback at the original recording rate across all eleven EuRoC sequences. Table~\ref{tab:runtime_per_sequence} reports the mean frontend and backend processing times for each sequence. Backend timing includes the visual correction and the associated IMU propagation. The final row gives the means over all recorded calls. Table~\ref{tab:runtime_configuration} summarizes the configured feature limits, the maximum observed numbers of active landmarks, and the corresponding Riccati matrix dimensions. For $n$ active landmarks, Riccati storage and structured propagation require $\mathcal{O}(n^2)$ memory and arithmetic operations, respectively. The computational cost of the visual correction increases with the number of active landmarks, with worst-case complexity $\mathcal{O}(n^3)$ when each landmark contributes a measurement.

\begin{table}[t]
    \centering
    \small
    \setlength{\tabcolsep}{4pt}
    \caption{Frontend and backend processing times
    (ms) on the EuRoC MAV dataset, averaged over
    recorded calls. Backend timing includes visual
    correction and the associated IMU propagation.
    The final row combines all eleven sequences.}
    \label{tab:runtime_per_sequence}
    \begin{tabular}{lrrrr}
        \toprule
        & \multicolumn{2}{c}{Stereo}
        & \multicolumn{2}{c}{Monocular} \\
        \cmidrule(lr){2-3}
        \cmidrule(lr){4-5}
        Sequence & Frontend & Backend
                 & Frontend & Backend \\
        \midrule
        MH\_01 & 16.80 & 21.43 & 16.27 & 14.79 \\
        MH\_02 & 17.07 & 19.47 & 16.01 & 14.15 \\
        MH\_03 & 16.58 & 21.73 & 15.05 & 14.22 \\
        MH\_04 & 14.78 & 19.28 & 13.63 & 11.68 \\
        MH\_05 & 15.82 & 20.32 & 13.02 & 13.32 \\
        \midrule
        V1\_01 & 14.36 & 22.34 & 12.79 & 16.49 \\
        V1\_02 & 14.74 & 18.55 & 11.73 & 10.45 \\
        V1\_03 & 14.16 & 15.99 & 10.26 & 6.07 \\
        \midrule
        V2\_01 & 14.97 & 19.53 & 12.47 & 13.58 \\
        V2\_02 & 14.88 & 17.39 & 11.69 & 9.56 \\
        V2\_03 & 13.59 & 7.77 & 10.39 & 6.01 \\
        \midrule
        Mean & 15.42 & 18.99 & 13.37 & 12.34 \\
        \bottomrule
    \end{tabular}
\end{table}

\begin{table}[t]
    \centering
    \small
    \caption{Feature limits and maximum observed
    observer sizes in the timing experiments.
    Here $n$ counts active landmarks.}
    \label{tab:runtime_configuration}
    \begin{tabular}{lrr}
        \toprule
        Parameter & Stereo & Monocular \\
        \midrule
        %Number of cameras & 2 & 1 \\
        Frontend feature cap & 80 & 80 \\
        Maximum observed $n$ & 72 & 70 \\
        Maximum Riccati matrix size
        & $222\times222$ & $216\times216$ \\
        \bottomrule
    \end{tabular}
\end{table}

\section{CONCLUSIONS}\label{s8}
This paper introduced a novel nonlinear geometric observer for visual--inertial odometry, formulated on the gravity-augmented matrix Lie group \(\mathbb{SE}_{3+n}(3)\), which describes pose, gravity, linear velocity, and landmark positions within a consistent geometric framework. Under the bias-free IMU model, the observable error equilibrium is almost globally asymptotically stable for relative-position, stereo-bearing, or monocular-bearing measurements, provided the corresponding observability conditions hold. A key feature of the design lies in the use of a geometric error definition that naturally leads to a time-varying linear translational subsystem, thereby eliminating the need for local linearization. The inverse Riccati matrix defines a quadratic Lyapunov function for this translational error system. Moreover, by decoupling the rotational and translational dynamics through the introduction of auxiliary dynamics, the observer adopts a cascaded structure in which the translational subsystem admits a Riccati-based (globally exponentially stable) estimator, while the rotational subsystem achieves almost global input-to-state stability. Global translation and rotation about gravity remain unobservable. Monte Carlo simulations illustrate the effect of measurement noise and visual outliers, and experiments on all eleven sequences of the \textit{EuRoC MAV} dataset quantify the position accuracy of the stereo and monocular observers. Future research will explore extensions of the proposed framework to account for IMU biases and to further assess its performance on additional real-world datasets. Additional efforts will also focus on deploying the proposed observer on embedded robotic platforms.

%Monte Carlo simulations validated the theoretical results, confirming the proposed observer’s effectiveness and resilience to sensor noise and initialization uncertainty. Future research will explore the extension of this framework to account for time-varying IMU biases and to validate its performance on real-world datasets. Additional efforts will also focus on deploying the proposed observer on embedded robotic platforms.

%\addtolength{\textheight}{-10cm}   % This command serves to balance the column lengths
                                  % on the last page of the document manually. It shortens
                                  % the textheight of the last page by a suitable amount.
                                  % This command does not take effect until the next page
                                  % so it should come on the page before the last. Make
                                  % sure that you do not shorten the textheight too much.

%%%%%%%%%%%%%%%%%%%%%%%%%%%%%%%%%%%%%%%%%%%%%%%%%%%%%%%%%%%%%%%%%%%%%%%%%%%%%%%%

%%%%%%%%%%%%%%%%%%%%%%%%%%%%%%%%%%%%%%%%%%%%%%%%%%%%%%%%%%%%%%%%%%%%%%%%%%%%%%%%

\appendices 

\section{Proof of Lemma \ref{lem_uo}} \label{app_uo}
The proof of this Lemma is motivated by \cite{Wang_TAC2021}. Let $\bold T(t) = I_{n+2} \otimes R^\top(t)$, $\bold S(t)= I_{n+2} \otimes \left(-[\omega^\mathcal{B}]_\times\right)$ and $\bar{\bold A}=\bold A(t)-\bold S(t)$. Note that the matrix $\bar{\bold A}$ is a constant matrix and is given as $\bar{\bold A}=A \otimes I_3$, where $A$ is defined as follows
\begin{align}
    A&=\begin{bmatrix}
       D&0_{2\times n}\\
        B_n&0_{n \times n}\\
    \end{bmatrix}.
\end{align}
with $D=\begin{bmatrix}
        0&1\\
        0&0
    \end{bmatrix}$ and $B_n$ is given in \eqref{B_definition}. One can check that $\dot{\bold T}(t)=\bold S(t) \bold T(t)$ and $\bold T(t) \bar{\bold A}=\bar{\bold A}\bold T(t)$. Now, we will demonstrate that the state transition matrix \(\bold{\Phi}(t, \tau)\) can be expressed as 
\begin{equation}\label{dt_phi}
    \bold\Phi(t, \tau) = \bold T(t) \bar{\bold\Phi}(t, \tau) \bold T^{-1}(\tau),
\end{equation} 
where \(\bar{\bold\Phi}(t, \tau) = \exp(\bar{\bold A}(t-\tau))\) is the state transition matrix of \(\bar{\bold A}\). Note that \(\bold\Phi(t, t) = I_{3n+6}\), \(\bold\Phi^{-1}(t, \tau)=\bold\Phi(\tau, t)\) and \(\bold\Phi(t_3, t_2) \bold\Phi(t_2, t_1) = \bold\Phi(t_3, t_1)\) for any \(t_1, t_2, t_3 \geq 0\). To verify that \(\bold\Phi(t, \tau)\) is the state transition matrix associated with \(\bold A(t)\), we will find its derivative with respect to \(t\). It follows from \eqref{dt_phi} that
\[
\frac{d}{dt} \bold \Phi(t, \tau) = \dot{\bold T}(t) \bar{\bold\Phi}(t, \tau) \bold T^{-1}(\tau) + \bold T(t) \frac{d}{dt}\bar{\bold \Phi}(t, \tau) \bold T^{-1}(\tau).
\]  
Using the fact that \(\frac{d}{dt}\bar{\bold\Phi}(t, \tau) = \bar{\bold A}\bar{\bold \Phi}(t, \tau)\), $\dot{\bold T}(t)=\bold S(t) \bold T(t)$ and $\bold T(t) \bar{\bold A}=\bar{\bold A}\bold T(t)$, one has  
\[
\frac{d}{dt} \bold\Phi(t, \tau) = \bold A(t) \bold \Phi(t, \tau).
\]  
This implies that \(\bold\Phi(t, \tau)\) is the state transition matrix associated with \(\bold A(t)\). Next, we establish the uniform observability of the pair \((\bold A(t), \bold C(t))\) by demonstrating the existence of constants $\delta>0$ and $\mu>0$ such that the observability Gramian satisfies the following condition
\begin{align}\label{og_1}
   \bold G_\delta(t)= \frac{1}{\delta} \int_t^{t + \delta} \bold\Phi(\tau, t)^\top \bold C(\tau)^\top \bold C(\tau) \bold \Phi(\tau, t) \, d\tau \geq \mu I_{3n+6}.
\end{align}
To facilitate the analysis, the output matrix $\boldsymbol C(t)$ is expressed as
\[
\boldsymbol C(t) = \boldsymbol P(t)\,\bar{\boldsymbol C},
\]
where $\boldsymbol P(t) = \operatorname{blkdiag}(\Pi_1,\dots,\Pi_n)$ and $\bar{\boldsymbol C}
= \begin{bmatrix}
0_{3n\times 6} & I_{3n}
\end{bmatrix}$. Invoking~\eqref{dt_phi}, together with the identities
$\boldsymbol T^{-1}(\tau) = \boldsymbol T^\top(\tau)$ and
\[
\bar{\boldsymbol C}\,\boldsymbol T(\tau)
= \left(I_n \otimes R^\top(\tau)\right)\bar{\boldsymbol C},
\]
the observability Gramian in~\eqref{og_1} can be rewritten as
    
{\small
\begin{align}\label{og_2}
    &\bold G_\delta(t)=\frac{1}{\delta} \int_t^{t + \delta} \bold\Phi(\tau, t)^\top \bold C(\tau)^\top \bold C(\tau) \bold \Phi(\tau, t) \, d\tau \nonumber\\
    &=\frac{1}{\delta} \bold T(t) \left( \int_t^{t + \delta} \bar{\bold\Phi}(\tau, t)^\top \bar{\bold C}^\top \bar{\bold P}(\tau)^\top \bar{\bold P}(\tau) \bar{\bold C}\bar{\bold \Phi}(\tau, t) \, d\tau \right) \bold T^{-1}(t), 
\end{align}}where $\bar{\bold P}(t)=(I_n \otimes R(t))\, \bold P(t) \, (I_n \otimes R^\top(t))$. Note that $\bar{\bold{P}}(t) = \operatorname{blkdiag}(\bar{\Pi}_1, \dots, \bar{\Pi}_n)$, where each $\bar{\Pi}_i$ is given by $\bar{\Pi}_i = R \, \Pi_i\, R^\top$ for $i = 1, 2, \dots, n$. To proceed with the proof, we analyze equation~\eqref{og_2} separately for each measurement type, as carried out in the sequel.

\subsubsection*{Relative position measurements} Considering the landmark measurements provided in \eqref{equ:measurements_3d_on_group}, one can verify that $\bar{\bold P}(t)=I_{3n}$, which further implies that 
\begin{align}\label{gram}
     \bold G_\delta(t)=\frac{1}{\delta} \bold T(t) \left( \int_t^{t+\delta} \bar{\bold \Phi}^\top(\tau, t) \bar{\bold C}^\top \bar{\bold C} \bar{\bold \Phi}(\tau, t) \, d\tau \right) \bold T^{-1}(t).
\end{align}
Now, we are going to show that the pair $(\bar{\bold A}, \bar{\bold C})$ is  (Kalman) observable. Let $\mathcal{O}$ denote the observability matrix associated with the pair $(\bar{\bold A}, \bar{\bold C})$, which is given by
\begin{align}
    \mathcal{O}=\begin{bmatrix}
        I_n&0_{n\times 2}\\
        0_n&B_n\\
        0_n&\bar B_n\\
        0_n&0_{n\times 2}\\
        \vdots&\vdots\\
        0_n&0_{n\times 2}\\
    \end{bmatrix}\otimes I_3, ~~\text{where}~~\bar B_n=\begin{bmatrix}
        0_{n \times 1} & \mathbf 1_{n\times 1}
    \end{bmatrix}.\label{obs_3d}
\end{align}
Since $\det(\mathcal{O}^\top \mathcal{O}) \neq 0$, the observability matrix $\mathcal{O}$ is full rank, which implies that the pair $(\bar{\bold A}, \bar{\bold C})$ is  (Kalman) observable. Consequently, there exist positive constants $\bar \delta > 0$ and $\bar \mu > 0$ such that:
\begin{align}\label{uo_cst}
    \frac{1}{\bar \delta} \int_t^{t+\bar \delta} \bold{\bar{\Phi}}^\top(\tau, t) \bar{\bold C}^\top \bar{\bold C} \bold{\bar{\Phi}}(\tau, t) \, d\tau \geq \bar \mu I_{3n+6}.
\end{align}
Substituting this back into \eqref{gram} and choosing $\delta> \bar \delta$ and $0<\mu\leq \frac{\bar \delta \, \bar \mu}{\delta}$, one concludes that 
\begin{equation}
    \bold G_\delta(t) \geq \mu I_{3n+6},
\end{equation}
for all \(t \geq 0\). This implies that the pair \((\bold A(t), \bold C(t))\) satisfies the uniform observability condition under the measurement \eqref{equ:measurements_3d_on_group}.

\subsubsection*{Stereo-bearing measurements} For the stereo-bearing measurements \eqref{equ:measurements_sb_on_group}, one can show that $\Pi_i$ is uniformly positive definite for each $i = 1, 2, \ldots, n$, which also implies that $\bar{\bold{P}}(t)$ is uniformly positive definite, \ie, $\bar{\bold{P}}(t) \geq \mu_{\bar{\bold{P}}}\,I_{3n}$ with $\mu_{\bar{\bold{P}}}>0$. It follows from \eqref{og_2} that 

{\small
\begin{align}\label{uo_1}
    \bold G_\delta(t)\geq\frac{\mu_{\bar{\bold{P}}}^{2}}{\delta} \bold T(t) \left( \int_t^{t + \delta} \bar{\bold\Phi}(\tau, t)^\top \bar{\bold C}^\top \bar{\bold C}\bar{\bold \Phi}(\tau, t) \, d\tau \right) \bold T^{-1}(t). 
\end{align}}Moreover, from \eqref{uo_cst} and \eqref{uo_1}, one can conclude that 
\begin{align}\label{uo_3}
    \bold G_\delta(t) \geq \mu \, I_{3n+6},
\end{align}
where $\delta > \bar \delta$ and $0 < \mu \leq \frac{\mu_{\bar{\bold{P}}}^{2}\, \bar \delta \, \bar \mu}{\delta}$. Inequality \eqref{uo_3} confirms that the pair $(\bold A(t), \bold C(t))$ satisfies the uniform observability condition considering the stereo-bearing measurements \eqref{equ:measurements_sb_on_group}.

\subsubsection*{Monocular-bearing measurements} %Similarly to the stereo-bearing measurements, considering 
Given the monocular-bearing measurements \eqref{equ:measurements_mb_on_group}, where each measurement satisfies inequality \eqref{ineq_uo} with constants $\delta^*>0$ and $\mu^*>0$, for $t \geq 0$, one can verify that 
%there exist two constants $\delta^*>0$ and $\mu^*>0$ such that 
 \begin{align}\label{ineq_1}
     \frac{1}{\delta^*}\int_{t}^{t+\delta^*} \bar{\bold P}(\tau)d\tau \mathrel{\geq} \mu^*\, I_{3n},
 \end{align}
 for all $t \geq 0$. Consider inequality \eqref{ineq_1}, along with the facts that $\bar{\bold P}^\top \bar{\bold P} = \bar{\bold P}$, the matrix $\bar{\bold P}$ is positive semi-definite, and the pair $(\bar{\bold A}, \bar{\bold C})$ is (Kalman) observable. According to \cite[Lemma 2.7]{HAMEL2017137}, it follows that there exist two constants $\bar{\bar{ \delta}}> 0$ and $\bar{\bar{\mu}} > 0$ such that

\begin{align} \label{ineq_2}
   \frac{1}{\bar{\bar{\delta}}} \int_t^{t + \bar{\bar{\delta}}} \bar{\bold\Phi}(\tau, t)^\top \bar{\bold C}^\top \bar{\bold P}(\tau)^\top \bar{\bold P}(\tau) \bar{\bold C}\bar{\bold \Phi}(\tau, t) \, d\tau \geq \bar{\bar{ \mu}} I_{3n+6}. 
\end{align}
Furthermore, picking $\delta>0$ and $\mu > 0$ such that $\delta > \bar{\bar{\delta}}$ and $\mu \leq \frac{\bar{\bar{\mu}}\,\bar{\bar{\delta}}}{\delta}$, it follows from inequality \eqref{ineq_2} that 
\begin{align}
    \bold G_\delta(t) \geq \mu I_{3n+6},
\end{align}
This implies that the pair $(\bold A(t), \bold C(t))$ is uniformly observable under the monocular-bearing measurements \eqref{equ:measurements_mb_on_group}. This completes the proof.

\section{Proof of Theorem \ref{thm:main}} \label{app_main_thr}
Let us begin this proof by analyzing the stability properties of the translational subsystem \eqref{csc_sys2}. Consider the following real-valued function:
\begin{equation}
    \mathcal{V}_P(x)=x^\top P^{-1} x,
\end{equation}
where \(P(t)\) is the solution of the CRE \eqref{CRE}. Since the two matrices $Q(t)$ and $V(t)$ are uniformly positive definite, and the pair $(\bold A(t), \bold C)$ is uniformly observable, it follows from the Riccati bounds \eqref{eq:riccati_bounds} that \( P(t) \), $\forall t \geq 0$, satisfies \( p_m \,I_{3n+6} \leq P(t) \leq p_M \,I_{3n+6} \) for some constants \( 0 < p_m \leq p_M < \infty \). This implies that 
\begin{equation}
    \frac{1}{p_M}||x||^2 \leq \mathcal{V}_P(x)\leq \frac{1}{p_m}||x||^2.
\end{equation}
Considering the dynamics \eqref{csc_sys2}, the time-derivative of $\mathcal{V}_P$ is given by
\begin{align}\label{ineq_v}
    \dot{\mathcal{V}}_P=& x^\top(P^{-1} \bold A + \bold A^\top P^{-1}-2 \bold C^\top Q^{-1} \bold C -P^{-1} \dot P P^{-1}) x\nonumber\\
    =&-x^\top P^{-1} V P^{-1}\, x-x^\top \bold C^\top Q^{-1} \bold C \,x\nonumber\\
    \leq & -\frac{v_m}{p^2_M}||x||^2\leq -\rho \mathcal{V}_P,
\end{align}
with \(\rho = v_m \frac{p_m}{p^2_M}\) and $v_m = \inf_{t \geq 0} \lambda^{V(t)}_{\text{min}}$ where $\lambda^{V(t)}_{\text{min}}$ is the minimum eigenvalue of $V$ at $t$. The last inequalities were derived using the facts $\dot P^{-1}=-P^{-1} \dot P P^{-1}$ and \(\bold C^\top Q^{-1} \bold C \geq 0\). It follows from \eqref{ineq_v} that \(||x(t)||\leq \sqrt{\frac{p_M}{p_m}} \exp{\left(-\frac{\rho}{2}t\right)}||x(0)||\),  which implies that \(x\) converges to zero exponentially. Next, to establish the stability properties of the overall closed-loop system \eqref{csc_sys1}-\eqref{csc_sys2}, we will first derive two results considering the reduced attitude closed-loop system \eqref{csc_sys1}. The first result establishes the stability properties of the system \eqref{csc_sys1} with $x=0$, and the second studies the Input to State Stability (ISS) properties of \eqref{csc_sys1} with respect to $x$. Finally, we will conclude the stability properties of the overall system \eqref{csc_sys1}-\eqref{csc_sys2}. The following lemma represents the first result related to the system \eqref{csc_sys1} with $x=0$.
\begin{lem}\label{lem1}
    Let $k_R>0$. Then, the following statements hold:
     \begin{enumerate}[i)]
     \item System \eqref{csc_sys1}, with $x=0$, has two equilibria: the desired equilibrium $\breve g=g$ and the undesired one $\breve g=-g$.\label{eq_p}
     \item The desired equilibrium $\breve g=g$, for system \eqref{csc_sys1} with $x=0$, is almost globally asymptotically stable on $\mathbb{S}^2_g$.\label{des_eq}
     \end{enumerate}
\end{lem}
\begin{proof}
From \eqref{csc_sys1}, with $x=0$, it is clear that the equilibrium points on $\mathbb{S}^2_g$ are $\breve g=g$ and $\breve g=-g$. The time-derivative of the following positive definite function  $\mathcal{L}_1=\frac{1}{2}||g-\breve g||^2$, along the trajectories of \eqref{csc_sys1} with $x=0$, is given by  $\dot{\mathcal{L}}_1=-k_R||g \times \breve g||^2$. The time-derivative of the following positive definite function $\mathcal{L}_2=\frac{1}{2}||g+\breve g||^2$, along the trajectories of \eqref{csc_sys1} with $x=0$, is given by  $\dot{\mathcal{L}}_2=k_R||g \times \breve g||^2$. Therefore, the equilibrium point $\breve g=g$ is stable and the equilibrium point $\breve g=-g$ is unstable (repeller). Almost global asymptotic stability on $\mathbb{S}^2_g$ of $\breve g=g$ immediately follows. The time derivatives of $\mathcal{L}_1$ and $\mathcal{L}_2$ were obtained using the fact that $u^\top [v]_\times u=0$ and $u^\top[u \times v]_\times v =-||u \times v||^2$ for every $u, v \in \mathbb{R}^3$. This completes the proof of Lemma \ref{lem1}.
\end{proof}
Next, we will study the ISS property of \eqref{csc_sys1}, with respect to $x$ using the notion of almost global ISS introduced in \cite{Angeli_TAC2011}. 

\begin{lem}\label{lemma:ISS}
    Given $k_R>0$, the system \eqref{csc_sys1} is almost globally ISS with respect to the equilibrium point $\breve g =g$ and the input $x$.
\end{lem}

\begin{proof}
    From the fact that the system \eqref{csc_sys2} is globally exponentially stable, it follows that the state $x$ belongs to a compact set $\mathcal{A} \subset \mathbb{R}^{3(n+2)}$. This, together with the fact that $\breve g$ belongs to a compact manifold $\mathbb{S}_g^2$, ensures that the system \eqref{csc_sys1}, subject to the bounded inputs $x$, evolves on the compact manifold $\mathbb{S}_g^2\times\mathcal{A}$. Consequently, condition A0, given in \cite{Angeli_TAC2011}, is fulfilled. Moreover, considering the positive definite function  $\mathcal{L}_1$ and system \eqref{csc_sys1} with $x=0$, condition A1, given in \cite{Angeli_TAC2011}, is also fulfilled. To check condition A2, given in \cite{Angeli_TAC2011}, let us derive the first-order approximation of system \eqref{csc_sys1} with $x=0$. To do so, we apply a small perturbation to $\breve g$ around the undesired equilibrium $-g$ through the attitude estimation error $\tilde R$ by letting $\tilde{R} = \mathcal{R}(\pi, \bar{u}) \exp([\zeta]_\times)$, where $\bar{u} \perp g$ and $\zeta \in \mathbb{R}^3$ is sufficiently small. Using the approximation $\exp([\zeta]_\times) \approx I_3 + [\zeta]_\times$ for sufficiently small $\zeta$, we have $\breve{g} \mathrel{\approx} -g + [\zeta]_\times g$. Letting $z:= [\zeta]_\times g$, one can see $z$ as the small perturbation applied to $\breve g$ around the undesired equilibrium $-g$. By neglecting the cross terms and using the fact that $[u\times v]_\times =vu^\top-uv^\top$ for every $u, v \in \mathbb{R}^3$, one has
    \begin{align}
        \dot z &= k_R M z \label{lti_u},
    \end{align}
    where $M:=\text{tr}(g g^\top) I_3-gg^\top=\text{diag}\left(||g||^2, ||g||^2, 0\right)$. This shows that the linearized dynamics of \eqref{csc_sys1} with $x=0$ have two identical positive eigenvalues $k_R\,||g||^2$. Thus, condition A2 is also fulfilled. Now, consider the positive definite function  $\mathcal{L}_1$ whose time derivative along the trajectories of \eqref{csc_sys1} is given by
    \begin{align}
        \dot{\mathcal{L}}_1&=-k_R||g \times \breve g||^2\mathbin{+}g^\top[\breve g]_\times \Upsilon(t) x\nonumber\\
        &=-k_R||[g]_\times \left(g-\breve g\right)||^2\mathbin{+}g^\top[\breve g]_\times \Upsilon(t) x\nonumber\\
        &=-k_R ||g||^2||g-\breve g||^2\mathbin{+}k_R\Big(||g||^{4}-2||g||^2g^\top\breve g\nonumber\\
        &~~~~~~~~~~~~~~~+\left(\breve g^\top g\right)^2\Big)\mathbin{+}g^\top [\breve g]_\times \Upsilon(t) x.
    \end{align}
    The last equality was obtained using the fact that $[u]_\times^2=-||u||^2I_3+uu^\top$ for every $u \in \mathbb{R}^3$. Next, let $c_1$ and $c_2$ be constants such that $+k_R\big(||g||^{4}-2||g||^2g^\top\breve g+\left(\breve g^\top g\right)^2\big)\leq c_1$ and $||g^\top[\breve g]_\times \Upsilon(t)|| \leq c_2$. One has
    \begin{align}
        \dot{\mathcal{L}}_1\leq-2k_R||g||^2\mathcal{L}_1+c_1+c_2||x||. \label{equ_iss}
    \end{align}
    It follows from \eqref{equ_iss} that system \eqref{csc_sys1} satisfies the ultimate boundedness property introduced in \cite[Proposition 3]{Angeli_TAC2011}. Therefore, according to \cite[Proposition 2]{Angeli_TAC2011}, one can conclude that system \eqref{csc_sys1} is almost globally ISS with respect to $\breve g=g$ and the input $x$. This completes the
proof of Lemma \ref{lemma:ISS}.
\end{proof}
Since the equilibrium $x=0$ for the system \eqref{csc_sys2} is UGES and the system \eqref{csc_sys1} with $x=0$ is AGAS and almost globally ISS with respect to $x$, one can conclude that the cascaded system \eqref{csc_sys1}-\eqref{csc_sys2} is AGAS. This completes the proof of Theorem~\ref{thm:main}.

\bibliographystyle{IEEEtran}
\bibliography{References}

@article{Bucy1972,
  author  = {R. S. Bucy},
  title   = {The {Riccati} Equation and Its Bounds},
  journal = {Journal of Computer and System Sciences},
  volume  = {6},
  number  = {4},
  pages   = {343--353},
  year    = {1972},
  doi     = {10.1016/S0022-0000(72)80026-6}
}

@article{Usenko_RAL2020,
  author = {Usenko, Vladyslav and Demmel, Nikolaus and Schubert, David and St{\"u}ckler, J{\"o}rg and Cremers, Daniel},
  title = {Visual-inertial mapping with non-linear factor recovery},
  journal = {IEEE Robotics and Automation Letters},
  volume = {5},
  number = {2},
  year = {2020},
  pages = {422--429},
  doi = {10.1109/LRA.2019.2961227}
}

@inproceedings{Scaramuzza_IEEERSJ2018,
  author = {Zhang, Zichao and Scaramuzza, Davide},
  title = {A tutorial on quantitative trajectory evaluation for visual(-inertial) odometry},
  booktitle = {IEEE/RSJ International Conference on Intelligent Robots and Systems (IROS)},
  year = {2018},
  pages = {7244--7251},
  doi = {10.1109/IROS.2018.8593941}
}

@inproceedings{Scaramuzza_ICRA2018,
  author = {Delmerico, Jeffrey and Scaramuzza, Davide},
  title = {A benchmark comparison of monocular visual-inertial odometry algorithms for flying robots},
  booktitle = {IEEE International Conference on Robotics and Automation (ICRA)},
  year = {2018},
  pages = {2502--2509},
  doi = {10.1109/ICRA.2018.8460664}
}

@article{Brossard_SJ2019,
  author = {Brossard, Martin and Barrau, Axel and Bonnabel, Silvère},
  title = {Exploiting symmetries to design {EKFs} with consistency properties for navigation and {SLAM}},
  journal = {IEEE Sensors Journal},
  volume = {19},
  number = {4},
  year = {2019},
  pages = {1572--1579}
}

@book{Chen1999,
  author = {Chen, Chi-Tsong},
  title = {Linear System Theory and Design},
  edition = {3rd},
  publisher = {Oxford University Press},
  address = {New York, NY},
  year = {1999}
}

@article{EuRoC_dataset,
  author = {Burri, Michael and Nikolic, Janosch and Gohl, Pascal and Schneider, Thomas and Rehder, J{\"o}rn and Omari, Sammy and Achtelik, Markus and Siegwart, Roland Y.},
  title = {The {EuRoC} micro aerial vehicle datasets},
  journal = {The International Journal of Robotics Research},
  volume = {35},
  year = {2016},
  pages = {1157--1163},
  doi = {10.1177/0278364915620033},
  url = {https://doi.org/10.1177/0278364915620033}
}

@article{Markdahl_SCL2017,
  author = {Markdahl, Johan and Hoppe, Jens and Wang, Lin and Hu, Xiaoming},
  title = {A geodesic feedback law to decouple the full and reduced attitude},
  journal = {Systems \& Control Letters},
  volume = {102},
  year = {2017},
  pages = {32--41},
  issn = {0167-6911}
}

@misc{Lim2020ARO,
  author = {Lim, Kai Li and Br{\"a}unl, Thomas},
  title = {A review of visual odometry methods and its applications for autonomous driving},
  howpublished = {arXiv:2009.09193},
  year = {2020}
}

@article{Gui_2015,
  author = {Gui, Jianjun and Gu, Dongbing and Wang, Sen and Hu, Huosheng},
  title = {A review of visual inertial odometry from filtering and optimisation perspectives},
  journal = {Advanced Robotics},
  publisher = {Taylor \& Francis},
  volume = {29},
  number = {20},
  year = {2015},
  pages = {1289--1301}
}

@inproceedings{Brossard_2018,
  author = {Brossard, Martin and Bonnabel, Silvère and Barrau, Axel},
  title = {Invariant {Kalman} filtering for visual inertial {SLAM}},
  booktitle = {International Conference on Information Fusion (FUSION)},
  year = {2018},
  pages = {2021--2028}
}

@article{Yang_RAL2022,
  author = {Yang, Yulin and Chen, Chuchu and Lee, Woosik and Huang, Guoquan},
  title = {Decoupled right invariant error states for consistent visual-inertial navigation},
  journal = {IEEE Robotics and Automation Letters},
  volume = {7},
  number = {2},
  year = {2022},
  pages = {1627--1634}
}

@inproceedings{Wu_IROS2017,
  author = {Wu, Kanzhi and Zhang, Teng and Su, Daobilige and Huang, Shoudong and Dissanayake, Gamini},
  title = {An invariant-{EKF} {VINS} algorithm for improving consistency},
  booktitle = {IEEE/RSJ International Conference on Intelligent Robots and Systems (IROS)},
  year = {2017},
  pages = {1578--1585}
}

@article{Zhang_RAL2017,
  author = {Zhang, Teng and Wu, Kanzhi and Song, Jingwei and Huang, Shoudong and Dissanayake, Gamini},
  title = {Convergence and consistency analysis for a {3-D} invariant-{EKF} {SLAM}},
  journal = {IEEE Robotics and Automation Letters},
  volume = {2},
  number = {2},
  year = {2017},
  pages = {733--740}
}

@article{Pieter_TAC2023,
  author = {van Goor, Pieter and Hamel, Tarek and Mahony, Robert},
  title = {Equivariant filter ({EqF})},
  journal = {IEEE Transactions on Automatic Control},
  volume = {68},
  number = {6},
  year = {2023},
  pages = {3501--3512}
}

@article{barrau2018invariant,
  author = {Barrau, Axel and Bonnabel, Silvère},
  title = {Invariant {Kalman} filtering},
  journal = {Annual Review of Control, Robotics, and Autonomous Systems},
  volume = {1},
  number = {1},
  year = {2018},
  pages = {237--257}
}

@article{Qin_TOR2018,
  author = {Qin, Tong and Li, Peiliang and Shen, Shaojie},
  title = {{VINS-Mono}: A robust and versatile monocular visual-inertial state estimator},
  journal = {IEEE Transactions on Robotics},
  volume = {34},
  number = {4},
  year = {2018},
  pages = {1004--1020}
}

@article{Leutenegger_OKVIS,
  author = {Leutenegger, Stefan and Lynen, Simon and Bosse, Michael and Siegwart, Roland and Furgale, Paul},
  title = {Keyframe-based visual–inertial odometry using nonlinear optimization},
  journal = {The International Journal of Robotics Research},
  publisher = {Sage Publications, Inc.},
  address = {USA},
  volume = {34},
  number = {3},
  month = mar,
  year = {2015},
  pages = {314--334},
  issn = {0278-3649},
  numpages = {21}
}

@inproceedings{Huang_ISER_2009,
  author = {Huang, Guoquan Paul and Mourikis, Anastasios I. and Roumeliotis, Stergios I.},
  title = {A first-estimates {Jacobian} {EKF} for improving {SLAM} consistency},
  booktitle = {International Symposium on Experimental Robotics (ISER)},
  year = {2009},
  pages = {373--382}
}

@article{Huang_IJRR2010,
  author = {Huang, Guoquan P. and Mourikis, Anastasios I. and Roumeliotis, Stergios I.},
  title = {Observability-based rules for designing consistent {EKF} {SLAM} estimators},
  journal = {The International Journal of Robotics Research},
  volume = {29},
  number = {5},
  year = {2010},
  pages = {502--528}
}

@inproceedings{Geneva_ICRA2020,
  author = {Geneva, Patrick and Eckenhoff, Kevin and Lee, Woosik and Yang, Yulin and Huang, Guoquan},
  title = {{OpenVINS}: A research platform for visual-inertial estimation},
  booktitle = {IEEE International Conference on Robotics and Automation (ICRA)},
  year = {2020},
  pages = {4666--4672}
}

@article{Bloesch_IJRR2017,
  author = {Bloesch, Michael and Burri, Michael and Omari, Sammy and Hutter, Marco and Siegwart, Roland},
  title = {Iterated extended {Kalman} filter based visual-inertial odometry using direct photometric feedback},
  journal = {The International Journal of Robotics Research},
  volume = {36},
  number = {10},
  year = {2017},
  pages = {1053--1072}
}

@inproceedings{Bloesch_IROS2015,
  author = {Bloesch, Michael and Omari, Sammy and Hutter, Marco and Siegwart, Roland},
  title = {Robust visual inertial odometry using a direct {EKF}-based approach},
  booktitle = {IEEE/RSJ International Conference on Intelligent Robots and Systems (IROS)},
  year = {2015},
  pages = {298--304}
}

@inproceedings{Huai_IROS2018,
  author = {Huai, Zheng and Huang, Guoquan},
  title = {Robocentric visual-inertial odometry},
  booktitle = {IEEE/RSJ International Conference on Intelligent Robots and Systems (IROS)},
  year = {2018},
  pages = {6319--6326}
}

@article{Huai_IJRR2022,
  author = {Huai, Zheng and Huang, Guoquan},
  title = {Robocentric visual–inertial odometry},
  journal = {The International Journal of Robotics Research},
  volume = {41},
  number = {7},
  year = {2022},
  pages = {667--689}
}

@inproceedings{Mourikis_ICRA2007,
  author = {Mourikis, Anastasios I. and Roumeliotis, Stergios I.},
  title = {A multi-state constraint {Kalman} filter for vision-aided inertial navigation},
  booktitle = {IEEE International Conference on Robotics and Automation (ICRA)},
  year = {2007},
  pages = {3565--3572},
  doi = {10.1109/ROBOT.2007.364024}
}

@article{Mingyang_2013,
  author = {Li, Mingyang and Mourikis, Anastasios I.},
  title = {High-precision, consistent {EKF}-based visual-inertial odometry},
  journal = {The International Journal of Robotics Research},
  volume = {32},
  number = {6},
  year = {2013},
  pages = {690--711}
}

@article{HAMEL2017137,
  author = {Hamel, Tarek and Samson, Claude},
  title = {Position estimation from direction or range measurements},
  journal = {Automatica},
  volume = {82},
  year = {2017},
  pages = {137--144},
  issn = {0005-1098}
}

@article{Angeli_TAC2011,
  author = {Angeli, David and Praly, Laurent},
  title = {Stability robustness in the presence of exponentially unstable isolated equilibria},
  journal = {IEEE Transactions on Automatic Control},
  volume = {56},
  number = {7},
  year = {2011},
  pages = {1582--1592},
  doi = {10.1109/TAC.2010.2091170}
}

@article{Mahony_TAC2008,
  author = {Mahony, Robert and Hamel, Tarek and Pflimlin, Jean-Michel},
  title = {Nonlinear complementary filters on the special orthogonal group},
  journal = {IEEE Transactions on Automatic Control},
  volume = {53},
  number = {5},
  year = {2008},
  pages = {1203--1218},
  doi = {10.1109/TAC.2008.923738}
}

@article{Wang_TAC2021,
  author = {Wang, Miaomiao and Berkane, Soulaimane and Tayebi, Abdelhamid},
  title = {Nonlinear observers design for vision-aided inertial navigation systems},
  journal = {IEEE Transactions on Automatic Control},
  volume = {67},
  number = {4},
  year = {2022},
  pages = {1853--1868}
}

@misc{barrau_arxiv2016,
  author = {Barrau, Axel and Bonnabel, Silvère},
  title = {An {EKF-SLAM} algorithm with consistency properties},
  howpublished = {arXiv:1510.06263},
  year = {2016}
}

@article{Mahony_AJC2021,
  author = {Mahony, Robert and Hamel, Tarek and Trumpf, Jochen},
  title = {An homogeneous space geometry for simultaneous localisation and mapping},
  journal = {Annual Reviews in Control},
  volume = {51},
  year = {2021},
  pages = {254--267},
  issn = {1367-5788}
}

@article{Pieter_TR2023,
  author = {van Goor, Pieter and Mahony, Robert},
  title = {{EqVIO}: An equivariant filter for visual-inertial odometry},
  journal = {IEEE Transactions on Robotics},
  volume = {39},
  number = {5},
  year = {2023},
  pages = {3567--3585}
}

@article{barrau2016invariant,
  author = {Barrau, Axel and Bonnabel, Silvère},
  title = {The invariant extended {Kalman} filter as a stable observer},
  journal = {IEEE Transactions on Automatic Control},
  publisher = {IEEE},
  volume = {62},
  number = {4},
  year = {2017},
  pages = {1797--1812}
}

@inproceedings{Mahony_SLAM_CDC2017,
  author = {Mahony, Robert and Hamel, Tarek},
  title = {A geometric nonlinear observer for simultaneous localisation and mapping},
  booktitle = {IEEE Conference on Decision and Control (CDC)},
  year = {2017},
  pages = {2408--2415}
}

@article{VANGOOR_aut_2021,
  author = {van Goor, Pieter and Mahony, Robert and Hamel, Tarek and Trumpf, Jochen},
  title = {Constructive observer design for visual simultaneous localisation and mapping},
  journal = {Automatica},
  volume = {132},
  year = {2021},
  pages = {109803},
  issn = {0005-1098}
}

\end{document}